\documentclass{birkjour}
\usepackage[english]{babel}
\usepackage[latin1]{inputenc}
\usepackage[T1]{fontenc}
\usepackage{amsmath}
\usepackage{amsfonts}
\usepackage{amssymb}
\usepackage{appendix}
\usepackage{enumitem}
\usepackage{dsfont}
\usepackage{cite}
\newtheorem{theorem}{Theorem}[section]
\newtheorem{lemma}[theorem]{Lemma}
\newtheorem{proposition}[theorem]{Proposition}
\newtheorem{corollary}[theorem]{Corollary}
\theoremstyle{definition}

\theoremstyle{remark}
\newtheorem*{remark}{Remark}
\newtheorem*{remarks}{Remarks}

\numberwithin{equation}{section}
\newcommand{\laml}{\Lambda}
\newcommand{\chiln}{\mathcal{X}_\Lambda^n}
\newcommand{\hln}{\mathcal{H}_\Lambda^n}
\newcommand{\x}{\textbf{x}}
\newcommand{\y}{\textbf{y}}
\newcommand{\z}{\textbf{z}}
\newcommand{\w}{\textbf{w}}

\usepackage{color}
\begin{document}
%
%
%
%
%
%
%
%
%
%
\title[Fock-space localization for hard-core particles in disorder]
 {Low-energy Fock-space localization\\ for attractive hard-core particles in disorder}
%
\author[V.~Beaud]{Vincent Beaud}
\address{%
Technische Universit\"at M\"unchen\\
Zentrum Mathematik\\
Boltzmannstr.~3\\
DE 85748 Garching b. M\"unchen\\
Germany}
\email{vincent.beaud@ma.tum.de}
%
\author[S.~Warzel]{Simone Warzel}
\address{%
Technische Universit\"at M\"unchen\\
Zentrum Mathematik\\
Boltzmannstr.~3\\
DE 85748 Garching b. M\"unchen\\
Germany}
\email{warzel@ma.tum.de}
\date{April 25, 2017}
%
%
\begin{abstract}
We study a one-dimensional quantum system with an arbitrary number of hard-core particles on the lattice, which are subject to a deterministic attractive interaction as well as a random potential. Our choice of interaction is suggested by the spectral analysis of the {XXZ} quantum spin chain. The main result concerns a version of high-disorder Fock-space localization expressed here in the configuration space of hard-core particles. The proof relies on an energetically motivated Combes-Thomas estimate and an effective one-particle analysis. As an application, we show the exponential decay of the two-point function in the infinite system uniformly in the particle number.
\end{abstract}
%
\maketitle
%
\section{Introduction}

Imbrie's works~\cite{Imbrie2016,ImbriePRL}  notwithstanding, complete mathematical proofs of many-body localization in the bulk of the many-particle spectrum remain a challenge. Much mathematical progress has been devoted to the understanding of localization in integrable systems~\cite{AizG,KP,HSS,ANSS16,SiW,ASSN,SeiW2016}. First proofs of ground-state localization for weakly interacting (non-integrable) Fermions subject to the Aubry-Andr\'e potential \cite{Mastropietro2015,Mastropietro2017} as well as within the Hartree-Fock approximation~\cite{Ducatez2016} are among the latest highlights. Interestingly, there is also a recent proof of quasi-localization in the Bose-Hubbard model without disorder, based on Nekhoroshev estimates \cite{BolsRoeck}.

The present paper returns to earlier attempts at addressing the localization problem, namely proofs of multi-particle localization~\cite{AiW2009,ChuSu2009}. In these articles (see also \cite{ChuSubook} and references therein), systems of $ n $ particles are proven to exhibit (strong dynamical) localization with a bound on the localization length that diverges with the particle number. A salient point of our study is to emphasize that, by adapting the techniques in the aforementioned works, the issue of a divergent localization length bound may be absent; in fact, whenever the interaction among particles is \emph{attractive}, thereby naturally forcing clustering of particles. This further allows us to energetically restrict to the low-energy regime defined by the two-cluster break-up. The many-particle localization estimates derived here suffice in particular to conclude exponential decay of the two-point function of low-energy eigenstates. Since such a scenario is relevant to the understanding of exponential decay of correlations in the XXZ quantum spin chain, we shall concentrate on the case of hard-core (spinless) particles with a translation-invariant, nearest-neighbor, attractive interaction on the one-dimensional lattice, or a subset thereof, $\laml:= \left[-L,L\right]\cap\mathbb{Z}$ with $ L \in \mathbb{N} $.

Though from a slightly different perspective, related results were proven independently in the recent preprint~\cite{Stolz}.
\subsection{Model and assumptions}

The configuration space of $ n  $ hard-core particles on $ \laml $ is identified with the set
\begin{equation*}
 	\chiln := \left\lbrace\, \x =\lbrace x_1,x_2,\ldots,x_n\rbrace\in\laml^n : x_1<x_2<\ldots <x_n\,\right\rbrace 
\end{equation*}
of all ordered $n$-tuples in $\Lambda$. A configuration $\x\in\chiln$ can equivalently be understood as a subset of $n$ sites in $\Lambda$. Two sites $u,v\in\Lambda$ are said to be neighboring if $\vert u-v\vert =1$ and similarly a site $u\in\Lambda$ is neighboring a subset $B\subset\Lambda$ if $\mathrm{dist}(u,B)=1$. We refer to a subset of neighboring sites in $\x$ as a \textit{cluster} if it is not neighboring any other site of $\x$. The set of configurations may thus be partitioned into a disjoint union,
\begin{equation*}
	\chiln  = \bigcup\limits_{k=1}^n \mathcal{C}_\Lambda^{(k)}\,,
\end{equation*}
of configurations with exactly $k$ clusters, $1\leqslant k\leqslant n$. The Hilbert space of $n$ hard-core particles on $\Lambda$ is then
\begin{equation*}
	\hln:=l^2\bigl(\chiln\bigr) =  \bigoplus\limits_{k=1}^n l^2\bigl(\mathcal{C}_\Lambda^{(k)}\bigr)\,,
\end{equation*}
with inner product $ \langle \cdot , \cdot \rangle $. An orthonormal basis  is given by $\lbrace\delta_\x\rbrace_{\x\in\chiln}$ where $\delta_\x(\y) = \delta_{\x,\y}$.

The class of Hamiltonians considered here consists of a kinetic hopping term, a hard-core attractive interaction and a random potential. The hopping of particles on $ \Lambda $ is modeled by the adjacency matrix, 
\begin{equation}\label{adjacency}
	A\delta_\x := \sum_{\substack{\y \in \chiln \\ d(\x,\y) = 1}} \delta_\y\,, 
\end{equation}
where the summation extends over all configurations $\y \in \chiln$ whose $\ell^1$-distance, 
\begin{equation}\label{eq:el1}
	d(\x,\y) := \sum_{j=1}^n \lvert x_j-y_j\rvert \,,
\end{equation}
to $\x$ is one. Due to the hard-core constraint, any translation-invariant, nearest-neighbor, attractive interaction among the particles  increases linearly with the number of clusters in the given configuration. Hence, for $ k $-cluster configurations, we set
\begin{equation}\label{potential}
	U\delta_\x := k\delta_\x\,,\qquad \x\in\mathcal{C}_\Lambda^{(k)} \,.
\end{equation}
The number $ k \equiv k(\x) $ equals half of the number of natural cluster boundaries in  $ \x \in C_\Lambda^{(k)} $. For $g>1$, $\lambda \geqslant 0$, we then define on $\hln$ the operator
\begin{equation}\label{hamiltonian}
	H_\Lambda := -A + 2g\,U + \lambda V\,,
\end{equation}
where $V$ is a random potential described below. It should be emphasized that the strength of the hopping term is of order $O(1)$ in this definition. These Hamiltonians are closely related to the XXZ spin chain in its Ising phase as detailed in Subsection~\ref{subsection_XXZ} below. It is assumed throughout the paper that:\\
\smallskip
\begin{itemize}
	\item[\textbf{A1}] The random potential is given in terms of a collection of iid random variables $\lbrace \omega(\alpha)\rbrace_{\alpha\in\Lambda}$ through
	\begin{equation}\label{random_potential}
		V\delta_\x = \left( \sum\limits_{j=1}^n \omega(x_j) \right) \, \delta_\x\,.
	\end{equation}
	\item[\textbf{A2}] The distribution of each $ \omega(\alpha) $ has a bounded density $ \varrho \in L^\infty(\mathbb{R})$ with compact support, $ {\rm supp}\, \varrho \subset [0,\omega_{\max}] $.
	\end{itemize}
\begin{remark}
	While the lower-boundedness of the random variables is essential, the upper-boundedness can be relaxed; this assumption is made here to keep the paper short. In particular, \textbf{A1} and \textbf{A2} imply all the assumptions in \cite{AiW2009}.
\end{remark}
\medskip
In the dynamics generated by $ H_\Lambda  $, clustering is energetically favored. In fact, let $\mathcal{P}^{(k)}$ stand for the orthogonal projection onto the subspace $l^2\bigl(\mathcal{C}_\Lambda^{(k)}\bigr)$ of exactly $ k $ clusters and $\mathcal{Q}^{(k)}$ for the orthogonal projection onto $\bigoplus_{j=k}^n l^2\bigl(\mathcal{C}_\Lambda^{(j)}\bigr)$, the subspace of at least $ k $ clusters. Note that
\begin{equation*}
	\mathds{1} = \mathcal{Q}^{(1)}\,,\qquad \mathcal{Q}^{(k)} = \mathcal{P}^{(k)} + \mathcal{Q}^{(k+1)}
		\quad (1\leqslant k\leqslant n-1)\,,\qquad \mathcal{Q}^{(n)} = \mathcal{P}^{(n)}\,.
\end{equation*} 
The following monotonically increasing lower bounds hold for restrictions $\mathcal{Q}^{(k)}H_\Lambda \mathcal{Q}^{(k)}$ of the Hamiltonian to sectors with at least $k$ clusters.
\begin{lemma}\label{lemma_thresholds}
	Let $H_\Lambda$ be as in \eqref{hamiltonian} with $ \lambda V \geqslant 0$. Then for all $1\leqslant k\leqslant n$:
	\begin{equation}
		\mathcal{Q}^{(k)} H_\Lambda \mathcal{Q}^{(k)}\geqslant 2k(g-1) \mathcal{Q}^{(k)} \, . 
	\end{equation}
\end{lemma}
The elementary proof is spelled out in Appendix~\ref{app_thresholds}.

\subsection{Main result}

Our main result will be formulated in terms of the \textit{configurational eigenfunction correlator} of $ H_\Lambda $ as introduced in~\cite{A94}. It is defined for any interval $ I \subset \mathbb{R} $ by
\begin{equation*}
	Q^{(n)}_\Lambda(\x,\y,I) := \sum_{E \in \sigma(H_\Lambda)\cap I}
		\left|\langle \delta_\x , P_{\{E\}}(H_\Lambda) \delta_\y \rangle \right| \, , 
\end{equation*}
where $ P_J(H) $ stands for the spectral eigenprojection of $ H $ associated with $ J \subset \mathbb{R} $. Note that $H_\Lambda$ has discrete spectrum $\sigma(H_\Lambda)$. In case of non-degenerate eigenvalues, the above definition hence coincides with the sum over eigenvalues $E \in  I $ of the product of normalized eigenvectors $|\Phi_E(\x)|\,|\Phi_E(\y)|$.  We recall two useful relations from \cite[Ch.~7]{AWbook}:
\begin{align}\label{eq:ECboundstriv}
	Q^{(n)}_\Lambda(\x,\y,I)  &\leqslant\, \sqrt{Q^{(n)}_\Lambda(\x,\x,I) \,  Q^{(n)}_\Lambda(\y,\y,I) }\,,\\
	Q^{(n)}_\Lambda(\x,\x,I) &=\,  \langle \delta_\x , P_{I}(H_\Lambda) \delta_\x \rangle \leqslant 1 \,.\notag
\end{align}
The correlator is thus bounded in terms of the local density of states, which is exponentially small in $n$  --- reflecting the fact that the potential energy (and hence the total energy) of any configuration is of order $\mathcal{O}(n) $. This is easily seen by an adaptation of the semigroup method in \cite[Lemma~4.6]{AiW2009}:
\begin{lemma}
	Let $ I = [0,\sup I] $ be a compact interval and  $ \lambda > 0 $. Under Assumptions~\textbf{A1}-\textbf{A2}, there exist constants $ C, c \in (0,\infty) $ such that for all $ n $, $ \Lambda $ and $\x\in\chiln$:
	\begin{equation}\label{eq:nodensity}
		\mathbb{E}\left[  Q^{(n)}_\Lambda(\x,\x,I) \right] \leqslant C e^{-c n } \, . 
	\end{equation}
\end{lemma}
\begin{proof}
	The spectral theorem and the bound $ \mathds{1}_I(E) \leqslant e^{t\sup I } e^{-tE} $, $ t > 0 $, on the indicator function $ \mathds{1}_I $ of $ I $ yield $ Q^{(n)}_\Lambda(\x,\x,I) \leqslant e^{t \sup I}\langle \delta_\x , e^{-t H_\Lambda }\delta_\x \rangle $. The latter may be expressed probabilistically using the Feynman-Kac formula \cite[Prop.~II.3.12]{CaLa}:
	\begin{equation*} 
		\langle \delta_\x , e^{-t H_\Lambda }\delta_\x \rangle 
			= \int \exp\left( -\int_0^t \lambda V(\y(s)) \mathrm{d}s -2g\int_0^t U(\y(s)) \mathrm{d}s \right)
			\nu_\Lambda^{(\x;t)}(\mathrm{d}\y) \,,
	\end{equation*}
	where $ \nu_\Lambda^{(\x;t)} $ is the measure generated by $  -A_\Lambda $ on paths $ \{\y(s)\}_{0\leqslant s\leqslant t}$ starting at and returning to $\x$ in time $t$. Jensen's inequality, $\exp( -\int_0^t \lambda V(\y(s)) \mathrm{d}s) \leqslant \int_0^t  \exp( -t \lambda V(\y(s)) ) \mathrm{d}s/t $, then yields:
	\begin{align*}
		\mathbb{E}\left[\langle \delta_\x , e^{-t H_\Lambda }\delta_\x \rangle\right]
			&\leqslant\, \mathbb{E}\left[ e^{-t \lambda V(\x) } \right] 
				\langle\delta_\x , e^{-t \left(  -A + 2g\,U \right)_\Lambda }\delta_\x \rangle\\
			&\leqslant\, \mathbb{E}\left[ e^{-t \lambda V(\x) } \right]
				= \left(  \int e^{-t\lambda \omega}  \varrho(\omega) \mathrm{d}\omega \right)^n \, . 
	\end{align*}
	The proof is concluded by noting that the Laplace transform of  $ \varrho $ on the right is strictly smaller than one for $ t \lambda > 0 $.  
\end{proof}

The main result of this paper is the following theorem, which provides a version of Fock-space localization~\cite{Basko20061126} for the present system of hard-core particles.
\begin{theorem}\label{theorem_exponential_cor}
	Let $ g > 1 $ and $ \mu_\textsc{t} > 0 $ be such that 
	\begin{equation}\label{energy_interval}
		E(g,\mu_\textsc{t} ) := 4 g - 12 e^{\mu_\textsc{t}} > 0 \, ,
	\end{equation}
	$ I\subset \left[0, E(g,\mu_\textsc{t}) \right) $ be a compact interval and $\mu \in (0,\mu_\textsc{t}) $. Under Assumptions~\textbf{A1}-\textbf{A2}, there exist constants  $\lambda_0 ,C \in(0,\infty) $ such that for all $ n \geqslant 2 $, $ \Lambda $, all $ \x, \y \in 	\chiln  $, and all $\lambda>\lambda_0$:
	\begin{equation}\label{exponential_fm_1}
		\mathbb{E}\bigl[\vert Q^{(n)}_\Lambda(\x,\y,I)\vert\bigr]\leqslant C  \, F_\mu(\x,\y)\,, 
	\end{equation}
	where 
	\begin{equation}\label{definition_F}
		F_\mu(\x,\y) := \left\{ \begin{array}{ll}
			\vphantom{\displaystyle\sum_x} e^{-\mu |x_1-y_1|} & \mbox{if $ \x,\y \in \mathcal{C}$,} \\
			\displaystyle\sum_{\w \in \mathcal{C} }  e^{-\mu \left( d(\x,\w) + |w_1-y_1| \right)}
				& \mbox{if $ \x \not \in \mathcal{C} $ and $\y \in \mathcal{C}$,} \\
			\displaystyle\sum_{\w , \textbf{v} \in \mathcal{C} }
				e^{-\mu \left(d(\x,\w)+ d(\textbf{v},\y)+  |w_1-v_1| \right)}
				& \mbox{if  $\x,\y \not \in \mathcal{C} $.}
		\end{array} \right.
	\end{equation}
	Here and henceforth, $\mathcal{C}\equiv \mathcal{C}^{(1)}_\Lambda$ abbreviates clustered configurations in $\chiln$.
\end{theorem}
The proof essentially treats clustered and non-clustered configurations separately. Non-clustered configurations are localized by an energetically motivated Combes-Thomas estimate found in Section \ref{Sec2}. Localization of clustered configurations follows in Section~\ref{section_Green_function} by a standard (one-particle) argument, which is carried out on the level of the many-particle Green functions. The proof of Theorem~\ref{theorem_exponential_cor} in  Section~\ref{section_fm_to_eq} finally relies on a relation between  the Green function and the eigenfunction correlator from~\cite{AiW2009}.
\bigskip\\
\noindent We continue with further remarks:
\begin{enumerate}
	\item The spectrum of $ H_\mathbb{Z} $  with $ \lambda = 0 $ on $ \mathcal{H}_\mathbb{Z}^n  $ decomposes into (delocalized) bands caused by the energetic separation of clusters. The lowest such interval, essentially induced by fully clustered configurations $ \mathcal{C} $, is referred to as the droplet band and given by 
 	\begin{align*}
 		\Delta(n) &:=\, 2 \sqrt{g^2-1}
 			\left[\frac{\cosh(\rho_g n)-1 }{\sinh(\rho_g n)}, \frac{\cosh(\rho_g n)-1 }{\sinh(\rho_g n) } \right]\\
 			&\subset\,  \big[ 2(g-1),2(g+1)\big] \, , 
	\end{align*}
	where $ \rho_g := \ln (g + \sqrt{g^2-1} )$, see~\cite{NS01,NSS07,Fisch,FS14}. Since the spectrum of $  H_\mathbb{Z} $ then contains $ \Delta(n)  + \lambda\,  {\rm supp} \varrho  $, the above localization statement is not void for sufficiently large $ g $.  In particular, if $ {\rm supp} \varrho = [0,\omega_{\max}]  $, the random potential $ V $ has arbitrarily large clearings in the infinite system so that $\Delta(n) \subset \sigma( H_\mathbb{Z}) $ almost surely.\\
	\indent Although the localization estimates are uniform in the particle number, the above theorem does not address localization for typical realizations if the particle number is proportional to the system's size. That is to say, the bottom of the spectrum is typically above $E(g,\mu_\textsc{t})$ for large $n,\Lambda$ with $n/\vert\Lambda\vert = \mathrm{const}>0$. In fact, by similar methods as in the proof of Lemma \ref{lemma_thresholds}, we have
	\begin{equation*}
		\inf\sigma\bigl(H_\Lambda\bigr)\geqslant 2(g-1)
			+ \min\bigl\lbrace 2(g-1), \lambda V_\mathrm{min}^\mathcal{C}\bigr\rbrace\,,
	\end{equation*}
	where $V_\mathrm{min}^\mathcal{C}:= \min_{\x\in\mathcal{C}}\sum_{i=1}^{n}\omega(x_i)$. A Chernoff bound now yields for any $E\geqslant 0$
	\begin{equation*}
		\mathbb{P}\bigl(\lambda V_\mathrm{min}^\mathcal{C}\leqslant E\bigr)
			\leqslant (\vert\Lambda\vert-n+1)\inf\limits_{t>0}
			e^{tE}\left(\int e^{-t\lambda\omega}\varrho(\omega)\mathrm{d}\omega\right)^n\,,
	\end{equation*}
	which vanishes in the limit $\vert\Lambda\vert\to\infty$ if $n$ is proportional to $\vert\Lambda\vert$.
	\item The (non-optimal) restriction~\eqref{energy_interval} imposed on the value of $ g $ stems from the corresponding technical condition~\eqref{condition_combes_thomas} in the Combes-Thomas estimate below.
	\item At first sight, the summations over clustered configurations in the definition of $ F_\mu $ may look frightening. However, their entropic contribution is low since all these sums (e.g.\ $ \w \in  \mathcal{C} $) are equivalent to a summation over \emph{one} variable (e.g.\ $ w_1 \in \Lambda $). Introducing the function 
	\begin{equation}\label{definition_dbar_app}
		\overline{d}(\x,\y) := \min\limits_{\w,\textbf{v}\in\mathcal{C}} 
			\bigl\lbrace d(\x,\w) + \vert w_1-v_1\vert + d(\textbf{v},\y)\bigr\rbrace\,,
	\end{equation}
	one may simplify~\eqref{exponential_fm_1} and write
	\begin{equation}\label{exponential_fma}
		\mathbb{E}\bigl[\vert Q^{(n)}_\Lambda(\x,\y,I)\vert\bigr]\leqslant C  \, e^{-\mu \overline{d}(\x,\y) }\, ,
	\end{equation}
	(cf.~Lemma~\ref{lemma_distance_dbar}, which collects some properties of $ \overline{d} $). Note that $\overline{d}$ is no distance function, as it is not definite and fails to satisfy a triangle inequality. Nonetheless, even for clustered configurations, the result is stronger than in \cite{AiW2009} since the Hausdorff distance between the sets defined by $\x$ and $\y$ is much smaller than $ \overline{d}(\x,\y) $. Moreover,  the constants  $ C , \mu \in( 0,\infty) $ do not depend on the particle number $ n $. It would be interesting to see whether the techniques outlined in this paper can be further combined with \cite{AiW2009} or \cite{ChuSu2009,ChuSubook} to show localization higher up in the spectrum of $ H_\Lambda $, where clusters  $\mathcal{C}_\Lambda^{(k)} $ with finite $ k $ appear. 
	\item Thanks to~\eqref{eq:ECboundstriv}, the three bounds~\eqref{eq:nodensity}, \eqref{exponential_fm_1} and \eqref{exponential_fma} may be combined at one's convenience, i.e., for any triple $ s_1,s_2,s_3 \geqslant 0 $ with $ s_1+s_2+s_3 = 1 $:
	\begin{equation}\label{eq:combined}
		\mathbb{E}\bigl[\vert Q^{(n)}_\Lambda(\x,\y,I)\vert\bigr]
			\leqslant C \, e^{-s_1 c n } \, e^{- s_2 \mu \overline{d}(\x,\y) }\, F_{s_3\mu}(\x,\y) \, . 
	\end{equation}
	\item By standard arguments \cite{AWbook}, the bound~\eqref{exponential_fm_1} implies dynamical localization 
	\begin{equation*}
 		\mathbb{E}\bigl[\sup_{t\in \mathbb{R}}
 			\left| \langle \delta_\x , e^{-it H_\mathbb{Z}} P_I(H_\mathbb{Z}) \delta_\y \rangle \right| \big]
 			\leqslant \mathbb{E}[| Q^{(n)}_\mathbb{Z}(\x,\y,I)|] \leqslant C F_\mu(\x,\y) \, .
 	\end{equation*} 
	For any fixed $\x \in\mathcal{X}^{n}_\mathbb{Z}$ and using Lemma~\ref{lemma_sum}, $F_\mu(\x,\y)$ is summable with respect to $\y\in\mathcal{X}_\mathbb{Z}^n$. Thus, by the RAGE Theorem \cite{AizG}, any spectrum of $ H_\mathbb{Z} $ on $\mathcal{H}^{n}_\mathbb{Z}$ within $I$ is entirely pure point for any $ n $.
	\item In the one-particle case, localization is known to occur in one dimension for short-range hopping also at arbitrarily weak disorder $ \lambda > 0 $, cf.\ \cite{GMP,CaLa,AWbook,Schenker} and references therein. The idea behind the present proof suggests that this also applies to the localization of the droplet as long as $ n $ stays much smaller than $\vert \Lambda \vert$.
	A proof of such a result is not immediate, though.
\end{enumerate}
\smallskip
A further virtue of the bound~\eqref{exponential_fm_1} resides in its summability over configurations $\x$ and $\y$ containing at least one particle in disjoint subsets $U, V \subset \Lambda $. The key  technical observation here is Lemma~\ref{lemma_sum}. Such quantities are relevant in the discussion of exponential clustering of many-particle eigenstates. More precisely, for $U,V\subset \Lambda$ two connected and disjoint subsets, one may be interested in the correlator
\begin{equation*}
	q_\Lambda^{(n)}(U,V,I) := \sum_{E \in \sigma(H_\Lambda)\cap I}
		\sum\limits_{\substack{\x\in\chiln\\ \x\cap U\neq\emptyset}}
		\sum\limits_{\substack{\y\in\chiln\\ \y\cap V\neq\emptyset}}
		|\langle \delta_{\x} , P_{\{E\}}(H_\Lambda) \delta_{\y}  \rangle | \, , 
\end{equation*}
which coincides with the respective sum over the eigenfunction correlator. An immediate consequence of the bounds~\eqref{eq:combined} is the following 
\begin{corollary} 
	In the set-up of Theorem~\ref{theorem_exponential_cor}, there exist constants $ C , c, \mu \in (0,\infty) $ such that for all $ n \geqslant 2 $, $ \Lambda$, and all connected and disjoint $U,V\subset \Lambda$:
	\begin{equation}\label{eq:1prdm}
 		\mathbb{E}\left[q_\Lambda^{(n)}(U,V,I)\right] \leqslant C \,  e^{-cn} \,\exp\left( -\mu\,\mathrm{dist}(U,V) \right) \,,
	\end{equation}
	where $\mathrm{dist}(U,V):=\min\lbrace\vert u-v\vert : u\in U,v\in V\rbrace$ denotes the distance between $U$ and $V$.
\end{corollary}
\begin{proof}
	The proof relies on~\eqref{eq:combined} and Lemma~\ref{lemma_distance_dbar}(vi), according to which $\overline{d}(\x,\y) \geqslant \mathrm{dist}(U,V) - (n-1)$. Combining  this with Lemma~\ref{summability_F}, we arrive at the bound
	\begin{equation*}
		\sum\limits_{\substack{\x\in\chiln\\ \x\cap U\neq\emptyset}}
			\sum\limits_{\substack{\y\in\chiln\\ \y\cap V\neq\emptyset}}
			\mathbb{E}\left[Q_\Lambda^{(n)}\bigl(\x, \y,I\bigr)\right] \leqslant C \, (n+1) \,
			e^{- n (s_1 c - s_2 \mu)} e^{-s_2 \mu\,\mathrm{dist}(U,V) } \, ,
	\end{equation*}
	where $ c , C, \mu \in (0,\infty) $ are independent of $ n $ and $ \Lambda $. Choosing $ s_1,s_2 \in (0,1) $ such that $s_1 c > s_2 \mu $ yields the result.
\end{proof}

In more physical terms, the above result guarantees in particular that, for any non-degenerate, normalized  $ n $-particle eigenvector $ \Phi_E^{(n)} $ of $ H_\mathbb{Z} $ with energies $ E < E(g,\mu_\textsc{t} )  $, the corresponding one-particle reduced density matrix decays exponentially, i.e., there are (non-random) constants $ c ,\mu \in (0,\infty) $ and a random variable $ A > 0 $ with finite mean $ \mathbb{E}\left[ A \right] < \infty $ such that for all $ n \geqslant 2 $ and all $ u,v \in \mathbb{Z} $:
\begin{equation}\label{eq:twopoint}
	\left|\langle \Phi_E^{(n)} , a_u^{*} a_v \Phi_E^{(n)}\rangle \right|
		\leqslant A\,e^{-cn}\,(1+|u|^2)\,\exp\left( -\mu  |u-v| \right) \,
\end{equation}
where $ a_u^*, a_v $ denote the particle creation and annihilation operators. This follows immediately from~\eqref{eq:1prdm}, see e.g.~\cite[Ch.~7.1]{AWbook}.

\subsection{Relation to the XXZ spin chain}\label{subsection_XXZ}

The class of Hamiltonians introduced in (\ref{hamiltonian}) is intimately related to the disordered XXZ spin chain in its Ising phase. This subsection aims at bridging the gap between the two formulations. The Hilbert space for a chain of $N$ spins $1/2$ is the $N$-fold tensor product of $\mathbb{C}^2$, denoted by
\begin{equation*}
	\mathcal{H}_\textsc{n}^\textsc{xxz} = \bigotimes\limits_{k=1}^{N}\mathbb{C}_k^2\,,
\end{equation*}
and endowed with the inner product on tensor product spaces. The indices $1\leqslant k\leqslant N$ merely identify the single factors in the product. The XXZ Hamiltonian without disorder,
\begin{equation*}
	H_\textsc{n}^\textsc{xxz} := \sum\limits_{k=1}^{N-1} h_{k,k+1} \, ,  
\end{equation*}
is the sum of nearest-neighbor interactions
\begin{equation*}
	 h_{k,k+1} := -\frac{1}{\Delta}\bigl(S_{k}^{x}\otimes S_{k+1}^{x} + S_{k}^{y}\otimes S_{k+1}^{y}\bigr)
		+ \bigl(\frac{1}{4}\mathds{1}_k\otimes\mathds{1}_{k+1}- S_{k}^{z}\otimes S_{k+1}^{z}\bigr)\,,
\end{equation*}
acting on $\mathbb{C}_k^2\otimes\mathbb{C}_{k+1}^2$ and extended by unity to $\mathcal{H}_\textsc{n}^\textsc{xxz}$. Here, $\Delta^{-1}$ is a real constant and $S^{x,y,z}$ are the spin matrices, normalized to have eigenvalues $\pm\frac{1}{2}$. The constant term $1/4\cdot\mathds{1}$ merely serves normalization purposes. The regime where $\Delta^{-1}=0$ corresponds to the (ferromagnetic) Ising model, while the Ising phase of the XXZ Hamiltonian is described  by $1>\Delta^{-1}>0$.

The dynamics generated by $ H_\textsc{n}^\textsc{xxz} $ conserves the total $z$-component of the spin. It thus suffices to consider restrictions $H_{n,\textsc{n}}^\textsc{xxz}$ to superselection sectors $\mathcal{H}_{n,\textsc{n}}^\textsc{xxz}$ with a constant number $n$ of, say, down-spins. This property persists upon addition of non-vanishing fields in the $z$-direction. In the sequel, we concentrate on the so-called \textit{droplet Hamiltonian} with disorder,
\begin{equation}\label{hamiltonian_XXZ_++}
	H_\textsc{n}^{++} := H_\textsc{n}^\textsc{xxz} + \frac{\gamma}{2}\bigl(\mathds{1}-S^{z}_1-S^{z}_\textsc{n}\bigr) 
		+\frac{\lambda}{2\Delta} \sum\limits_{k=1}^{N} \omega(k) \bigl(\frac{1}{2} \, \mathds{1}-S_k^{z}\bigr)\,,
\end{equation}
where $\gamma\geqslant 0$ and $\lbrace\omega(k)\rbrace_{k=1}^N$ is the given family of iid random variables.  The superscript $++$ reflects the fact that having up-spins at both boundary sites is energetically most favorable. 
The Hamiltonian (\ref{hamiltonian}) with $n$ hard-core particles on $\Lambda$ is, up to a multiplicative constant, unitarily equivalent to the XXZ Hamiltonian (\ref{hamiltonian_XXZ_++}) restricted to the sector with $n$ down-spins on $\Lambda$, and with constants set to $\Delta=g$ and $\gamma=1$.  
This unitary equivalence is the object of the following proposition. It provides a dictionary
which allows one to translate results for hard-core particles to the {XXZ} system. In particular, two-point functions such as in~\eqref{eq:twopoint} relate to spin correlation functions. 
\begin{proposition}\label{prop:XXZ}
	Let $H_{n,\Lambda}(g)\equiv H_\Lambda$ and $H_{n,\textsc{n}}^{++}(\Delta,\gamma)\equiv H_{n,\textsc{n}}^{++}$ be as in \eqref{hamiltonian} and \eqref{hamiltonian_XXZ_++}. Then, there exists a unitary operator $\mathcal{U}:\mathcal{H}_\Lambda^n\to\mathcal{H}_{n,\vert\Lambda\vert}^\textsc{xxz}$ such that
	\begin{equation}\label{equivalence_hamiltonians}
		\mathcal{U}\,  H_{n,\Lambda}(g) \, \mathcal{U}^{*} = 2g \, H_{n,\vert\Lambda\vert}^{++}(g,1)\,,
	\end{equation}
	where $\vert\Lambda\vert$ stands for the number of sites in $\Lambda$.
\end{proposition}
For the reader's convenience, a proof is given in Appendix \ref{app_XXZ_relation}.
\bigskip\\
\noindent Many-body localization (MBL) has received extensive attention from the physics community~\cite{NHuse,VoAlt,IRS17} including early on the case of the {XXZ} quantum spin chain  \cite{Prosen08}. 
The definition of the term MBL varies from Fock-space localization to local integrals of motions (LIOM's). 
It is an interesting question to further clarify the validity and
relations of these notions even in the above situation.

\section{Controlling cluster break-up}\label{Sec2}

The main technical difference between the situation covered by \cite{AiW2009} and the present set-up lies in the presence of an additional spectral threshold. The spectrum of the Hamiltonian restricted to at least two clusters, $\mathcal{Q}^{(2)}H_\Lambda \mathcal{Q}^{(2)}$, is energetically higher than the regime of interest defined by essentially  one cluster, cf.~Lemma~\ref{lemma_thresholds}. By the Combes-Thomas estimate stated below, the corresponding Green function thus decays deterministically in this regime. To avoid cluttered expressions, we henceforth use the two shorthand notations:
\begin{equation*}
	H_\Lambda^{(k)} := \mathcal{Q}^{(k)}H_\Lambda \mathcal{Q}^{(k)}\,,
		\qquad G_\Lambda^{(k)}(\x,\y;z):=\langle \delta_\x,\bigl(H_\Lambda^{(k)}-z\bigr)^{-1}\delta_\y\rangle\,.
\end{equation*}
The operator $H_\Lambda^{(k)}-z$ is implicitly understood to act on the subspace $\mathcal{Q}^{(k)}\mathcal{H}$. The main result of this section is the following theorem.
\begin{theorem}\label{theorem_combes_thomas}
	For any $g >1$, $\mu_\textsc{t}>0$ and $E\geqslant 0$ satisfying
	\begin{equation}\label{condition_combes_thomas}
		4g-E > 12\,e^{\mu_\textsc{t}}\,,
	\end{equation}
	there exists a constant $C_\textsc{t} \equiv C_\textsc{t}(g,\mu_\textsc{t},E) \in(0, \infty)$ such that for all $ n \geqslant 2 $, $ \Lambda $, $\lambda\geqslant 0$ and all $\x, \y\in \chiln \backslash  \mathcal{C} $:
	\begin{equation}\label{combes_thomas}
		\vert G^{(2)}_\Lambda(\x,\y; E) \vert \leqslant C_\textsc{t} e^{-\mu_\textsc{t}d(\x,\y)}\, . 
	\end{equation}
\end{theorem}

\begin{remarks}
	1.~For any  $2\leqslant k\leqslant n$, the condition (\ref{condition_combes_thomas}) entails
	\begin{equation}\label{condition_delta_k}
		\delta_k(E) := 2k(g-e^{\mu_\textsc{t}}) -E \, > 4k e^{\mu_\textsc{t}}\,.
	\end{equation}
	As an immediate consequence of Lemma \ref{lemma_thresholds}, this definition implies the following inequality on $\mathcal{Q}^{(k)}\mathcal{H}$:
	\begin{equation}\label{eq_threshold_k}
		H_\Lambda^{(k)} - E > \delta_k(E) + 2k(e^{\mu_\textsc{t}}-1)\,.
	\end{equation}
	\noindent 2.~The subsequent proof yields:
	\begin{equation*}
	C_\textsc{t} \equiv C_\textsc{t}(g,\mu_\textsc{t},E)= \frac{2}{\delta_2(E)}\left(1-\frac{8e^{\mu_\textsc{t}}}{\delta_2(E)}\right)^{-1}  \in \left( 2/\delta_2(E) , \infty\right) \,,
	\end{equation*}
	where the inclusion is a direct consequence of (\ref{condition_delta_k}). For $g\to\infty$ and fixed $(\mu_\textsc{t}, E )$ the bound $C_\textsc{t}$ tends to $0$.
	\medskip\\
	\noindent 3.~Energies $E$ satisfying (\ref{condition_combes_thomas}) are always below the two-cluster threshold $4(g-1)$, but may be chosen arbitrarily close to it, provided $g$ is sufficiently large. Namely, for $E<4\alpha(g-1)$ with $\alpha<1$, condition~(\ref{condition_combes_thomas}) is satisfied if $(1-\alpha)g +\alpha>3e^{\mu_\textsc{t}}$.
\end{remarks}
We follow the basic  idea for the standard Combes-Thomas bound in \cite{A94} augmented by an inductive analysis using the Schur complement formula on $ k $-cluster subspaces $ \mathcal{P}^{(k)} \hln $. We fix $ \y \in \chiln$  and consider the following (bounded and invertible) multiplication operator $ M_\y \delta_\x := e^{\mu_\textsc{t}d(\x,\,\y)}\delta_\x $ on $ \hln $, which commutes with all of the projections $\mathcal{P}^{(k)}, \mathcal{Q}^{(k)}$. 
\begin{lemma}\label{lem:CT}
	In the set-up of Theorem~\ref{theorem_combes_thomas}, for any $ j,k \in \{ 2, \dots , n \} $ and  any $ \y \in \chiln$:
	\begin{align}
		&\Vert \mathcal{P}^{(k)}M_\y R^{(l)}(E)M_\y^{-1} \mathcal{P}^{(j)}\Vert  
			\leqslant \frac{2}{\delta_l(E)} \,,  \qquad  l\in\lbrace j,k\rbrace \,, \label{CT1}\\
		&\Vert \mathcal{P}^{(k)} M_\y R^{(2)}(E) M_\y^{-1}\mathcal{P}^{(j)}\Vert \leqslant C_\textsc{t}\,, \label{CT2}
	\end{align}
	where $ R^{(k)}(E) := \big( H^{(k)}_\Lambda - E \big)^{-1} $. 
\end{lemma} 

\begin{proof}[Proof of Theorem~\ref{theorem_combes_thomas}] 
	We assume without loss of generality that $ \x \in \mathcal{C}^{(k)}_\Lambda $ and $ \y \in \mathcal{C}^{(j)}_\Lambda $. Since $G^{(2)}_\Lambda(\x,\y;E) \, e^{\mu_\textsc{t}d(\x,\y)}= \bigl\langle\delta_\x, \mathcal{P}^{(k)} M_\y R^{(2)}(E)  M_\y^{-1} \mathcal{P}^{(j)} \delta_\y\bigr\rangle $, the claim is immediate from~\eqref{CT2}.
\end{proof}
It hence remains to give a proof of Lemma \ref{lem:CT}.

\begin{proof}[Proof of Lemma~\ref{lem:CT}]
	Since $ \y $ is fixed throughout the proof, we drop it from the notation, $ M \equiv M_\y $. Note that (\ref{CT1}) is only non-trivial for $l=j\leqslant k$ or $l=k\leqslant j$. It is first proven for $k=j=l$, inductively from the maximal number of clusters down to $k = 2$. We spell out the argument in case the maximal number of clusters is $ k = n $, corresponding to $ n \leqslant L +1 $. If $ L \leqslant n $, a similar argument applies.

	\noindent\textit{Base case, $ k = n $.} The operator $ B^{(n)} := M H_\Lambda^{(n)} M^{-1}-H_\Lambda^{(n)}$ is bounded through the Schur bound: 
	\begin{align}\label{bound_bn}
		\Vert B^{(n)}\Vert &\leqslant \sqrt{\sup_\x\sum_{\x'}\vert B^{(n)}(\x,\x')\vert}
			\sqrt{  \sup_{\x'}\sum_{\x}\vert B^{(n)}(\x,\x')\vert } \notag \\
		&\leqslant \sup_\x\sum_{\x'}\vert H_\Lambda^{(n)}(\x,\x')\vert\bigl(e^{\mu_\textsc{t}d(\x,\x')}-1\bigr)
			\leqslant 2n \bigl(e^{\mu_\textsc{t}}-1\bigr)\, . 
	\end{align}
	The last inequality used that only neighboring configurations $\x,\x'$ contribute to the sum, and for such: $ |H_\Lambda^{(n)}(\x,\x')| = | \langle \delta_\x, H_\Lambda^{(n)} \delta_{\x'} \rangle | = 1 $. Moreover,  for any given $\x\in\mathcal{C}^{(n)}_\Lambda$, there are at most $2n$ neighboring $\x'$. By~\eqref{eq_threshold_k}, we thus have $ {\rm dist}(E,\sigma(H_\Lambda^{(n)}))- \| B^{(n)}\| \geqslant \delta_n(E) > 0 $ and hence the operator $ M ( H_\Lambda^{(n)} - E ) M^{-1} =  H_\Lambda^{(n)} - E +  B^{(n)} $ is invertible on $\mathcal{Q}^{(n)}\mathcal{H}$ with bounded inverse:
	\begin{align*}
		\left\|M R^{(n)}(E) M^{-1}\right\| &=\,  \left\|\left(  H_\Lambda^{(n)} - E + B^{(n)}\right)^{-1}\right\|\\
			&\leqslant\, \left( {\rm dist}(E,\sigma(H_\Lambda^{(n)}))- \| B^{(n)}\| \right)^{-1} 
			\leqslant \delta_n(E)^{-1} \, . 
	\end{align*}
	\medskip\\
	\noindent\textit{Induction step, $ k+1 \to k $.} Assume that~(\ref{CT1}) holds for $j+1=k+1\leqslant n$. Using the orthogonal decomposition of $\mathcal{Q}^{(k)}$ into $\mathcal{P}^{(k)} + \mathcal{Q}^{(k+1)}$, the Schur complement formula yields
	\begin{align*}
		\mathcal{P}^{(k)} M R^{(k)}&(E) M^{-1} \mathcal{P}^{(k)}\\
			&=\, \mathcal{P}^{(k)} M \bigl(\mathcal{P}^{(k)}(H_\Lambda-E)\mathcal{P}^{(k)} - S^{(k)}(E)\bigr)^{-1}
				M^{-1}\mathcal{P}^{(k)} \notag \\
			&=\, \mathcal{P}^{(k)} \bigl(\mathcal{P}^{(k)}(H_\Lambda-E)\mathcal{P}^{(k)} + B^{(k)}(E)\bigr)^{-1}
				\mathcal{P}^{(k)}\,,
	\end{align*}
	where 
	\begin{align}
		S^{(k)}(E) &:=\, T_k R^{(k+1)}(E) T_k^{*}\,,
			\quad\textnormal{with}\quad T_k := \mathcal{P}^{(k)} H_\Lambda \mathcal{P}^{(k+1)}\,, \notag \\
		B^{(k)}(E) &:=\, M \mathcal{P}^{(k)}H_\Lambda \mathcal{P}^{(k)}M^{-1} - \mathcal{P}^{(k)}H_\Lambda \mathcal{P}^{(k)}
			- M S_\Lambda^{(k)}(E) M^{-1}\,.
	\label{definition_bk}
	\end{align}
	In the definition of $S^{(k)}(E)$, the operator $T_k$ and its adjoint $T_k^{*}$ arise from $\mathcal{P}^{(k)}H_\Lambda \mathcal{Q}^{(k+1)} = \mathcal{P}^{(k)}H_\Lambda \mathcal{P}^{(k+1)}$, since the hopping induced by $H_\Lambda$ only connects configurations whose number of clusters are at most one apart.
	
	By (\ref{eq_threshold_k}), we have $\mathcal{P}^{(k)}(H_\Lambda-E)\mathcal{P}^{(k)} \geqslant \delta_k(E) + 2k(e^{\mu_\textsc{t}}-1)$ on $\mathcal{P}^{(k)}\mathcal{H}$. The claimed bound now follows (by reasoning analogous to the base case) if $\delta_k(E) + 2k(e^{\mu_\textsc{t}}-1) -\Vert B^{(k)}(E)\Vert \geqslant \delta_k(E)/2 $, or equivalently, $\Vert B^{(k)}(E)\Vert \leqslant 2k(e^{\mu_\textsc{t}}-1) +\delta_k(E)/2 $. In fact, splitting (\ref{definition_bk}) into two additive terms and arguing as in (\ref{bound_bn}), we have on the one hand,
	\begin{align*}
		\bigl\Vert M \mathcal{P}^{(k)}H_\Lambda \mathcal{P}^{(k)}M^{-1}-\mathcal{P}^{(k)}H_\Lambda \mathcal{P}^{(k)}\bigr\Vert  
			&\leqslant\, \sup_\x\sum_{\x'}\vert H_\Lambda^{(k)}(\x,\x')\vert\bigl(e^{\mu_\textsc{t}d(\x,\x')}-1\bigr)\\
			&\leqslant\,  2(k-1)(e^{\mu_\textsc{t}}-1)\,,
	\end{align*}
	and on the other hand,
	\begin{equation}\label{eq:SchurH}
		\bigl\Vert M T_k M^{-1}\bigr\Vert \, , \; \bigl\Vert M T_k^* M^{-1}\bigr\Vert \leqslant 2k e^{\mu_\textsc{t}}\,. 
	\end{equation} 
	Together with the induction hypothesis this guarantees
	\begin{equation*}
		\bigl\Vert M S_\Lambda^{(k)}(E)M^{-1}\bigr\Vert \leqslant (2k e^{\mu_\textsc{t}})^2 \frac{2}{\delta_{k+1}(E)}
			\leqslant \frac{(4k e^{\mu_\textsc{t}})^2}{2\delta_{k+1}(E)} 
			< \frac{\delta_k(E)}{2}\, , 
	\end{equation*}
	where the last inequality is by (\ref{condition_delta_k}). This concludes the proof of~\eqref{CT1} for $j=k$.
	\bigskip\\	
	\noindent In the remainder of the proof, we assume $( 2 \leqslant)\,j \leqslant k$. The other case follows analogously. The arguments are by iteration, based on the following resolvent formula 
	for $ (2 \leqslant) \, m < k\, $:
	\begin{equation}\label{eq:reseqCT}
		\mathcal{P}^{(k)} R^{(m)}(E)
			=\mathcal{P}^{(k)} R^{(k)}(E)-\mathcal{P}^{(k)} R^{(k)}(E)\, T_{k-1}^* \,\mathcal{P}^{(k-1)}\, R^{(m)}(E) \, , 
	\end{equation} 
	where we used that
	\begin{equation*}
		\mathcal{Q}^{(k)} \left( H_\Lambda^{(m)} - H_\Lambda^{(k)} \right)
			= \mathcal{Q}^{(k)} H_\Lambda \left( \mathcal{Q}^{(m)} - \mathcal{Q}^{(k)} \right)
			=  \sum_{j=m}^{k-1} \mathcal{Q}^{(k)} H_\Lambda \mathcal{P}^{(j)} =  T_{k-1}^*\,.
	\end{equation*}
	For a proof of~\eqref{CT1} for $ 2 \leqslant l= j < k $, we set $m=j$ in (\ref{eq:reseqCT}) and note that the first term on the right hand side does not contribute. Together with the Schur bound~\eqref{eq:SchurH}, (\ref{CT1}) with $j=k$ and~\eqref{condition_delta_k}, this yields: 
	\begin{align}
		\bigl\Vert\mathcal{P}^{(k)} M R^{(j)}(E)M^{-1} &\mathcal{P}^{(j )}  \bigr\Vert 
			=\bigl\Vert\mathcal{P}^{(k)}M R^{(k)}(E)T^{*}_{k-1}\mathcal{P}^{(k-1)}R^{(j)}(E)M^{-1}\mathcal{P}^{(j)}\bigr\Vert 
				\notag\\
			&\leqslant\,\frac{4(k-1) e^{\mu_\textsc{t}}}{\delta_{k-1}(E)} \frac{\delta_{k-1}(E)}{\delta_{k}(E)} 
	   			\bigl\Vert\mathcal{P}^{(k-1)} M R^{(j)}(E)M^{-1}  \mathcal{P}^{(j)}\bigr\Vert\notag\\
			&\leqslant\, r  \, \bigl\Vert \mathcal{P}^{(k-1)} M R^{(j)}(E)M^{-1}  \mathcal{P}^{(j )}  \bigr\Vert \,,
	 			\qquad  r:=  \frac{8 e^{\mu_\textsc{t}}}{\delta_{2}(E)} < 1 \, .\notag 
	 \end{align}
	 Since $r<1$, iteration reduces the claim to the case $ k = j $ and thus concludes the proof of~\eqref{CT1}. 
	 
	 For a proof of~\eqref{CT2}, we only consider $j,k\geqslant 3$, since the other cases are covered by~\eqref{CT1}, and again assume $j\leqslant k$. By the adjoint of the resolvent formula (\ref{eq:reseqCT}) with $k=j$ and $m=2$, combined with (\ref{CT1}) and the Schur bound~(\ref{eq:SchurH}), we obtain
	 \begin{align*}
	 	\bigl\Vert\mathcal{P}^{(k)} M R^{(2)}&(E)M^{-1} \mathcal{P}^{(j)}\bigr\Vert\\
	 		&\leqslant\, \frac{2}{\delta_j(E)}
	 			+\frac{4(j-1) e^{\mu_\textsc{t}}}{\delta_{j-1}(E)} \frac{\delta_{j-1}(E)}{\delta_{j}(E)}
	 			\bigl\Vert\mathcal{P}^{(k)} M R^{(2)}(E)M^{-1}  \mathcal{P}^{(j-1)}\bigr\Vert\notag \\
	 		&\leqslant\,\frac{2}{\delta_2(E)}
	 			+r\,\bigl\Vert\mathcal{P}^{(k)} M R^{(2)}(E)M^{-1} \mathcal{P}^{(j-1)}\bigr\Vert \,.
	 \end{align*}
	 Iteration until $j-1=2$ yields the bound $\bigl(2/\delta_2(E)\bigr)(1-r)^{-1}$.
\end{proof}

\section{Bound on the Green function's fractional moments}\label{section_Green_function}

The main result, Theorem~\ref{theorem_exponential_cor}, will follow by proving bounds on fractional moments of the Green function, 
\begin{equation*}
	G_\Lambda(\x,\y;z) :=\langle \delta_\x, \bigl(H_\Lambda-z\bigr)^{-1}\delta_\y\rangle\,.
\end{equation*}
Our assumptions guarantee that such moments are finite. In fact, integration over just (one or) two variables associated with sites $u, v$  suffices provided these sites belong to the considered configurations (a fact, which we denote by  $u\in\x$ and $v\in\y$). In the following such conditional expectations will be denoted by $  \mathbb{E}\bigl[ \cdot \:\vert\:\omega_{\neq \lbrace u,v\rbrace}\bigr] := \int(\cdot)\varrho(\omega(u))\,\varrho(\omega(v))\mathrm{d}\omega(u)\,\mathrm{d}\omega(v)$. 
\begin{proposition}[Theorem~2.1 in \cite{AiW2009}] 
	Let $s \in (0,1)$. Under Assumptions \textbf{A1}-\textbf{A2}, there exists a constant $ C_s < \infty $ such that for all $n\in\mathbb{N}$, $\Lambda$, and all $\x,\y\in\chiln$ with $u\in\x$ and $v\in\y$:
	\begin{equation}\label{assumption_a2}
		\sup_{z \in \mathbb{C} } \; \mathbb{E}\bigl[\vert G_\Lambda(\x,\y;z)\vert^s \:\vert\: 
			\omega_{\neq \lbrace u,v\rbrace}\bigr]\leqslant \frac{C_s}{\lambda^s}\,.
	\end{equation}
\end{proposition}
The Combes-Thomas estimate shown in the previous section yields a deterministic exponential bound on $ G^{(2)}_\Lambda$, i.e.~when restricted to $\mathrm{ran}(\mathcal{Q}^{(2)})$. What is left is to localize the clusters. This is done by adopting a standard high-disorder localization technique for one-particle models using geometric decoupling.
\begin{theorem}\label{theorem_exponential_fm}
	Let $I\subset \left[0, E(g,\mu_\textsc{t}) \right)$ be a compact interval, $\overline{d}$ the function defined in \eqref{definition_dbar_app}, $s\in (0,1)$ and $\mu\in(0,\mu_\textsc{t})$. Under Assumptions~\textbf{A1}-\textbf{A2}, there exist constants $C_s^{(1)}, C_s^{(2)} \in (0,\infty) $ and $\lambda_0 >0$  such that for all $n\geqslant 2$, $\Lambda $, and all $\lambda>\lambda_0$, $E\in I$:
	\begin{align}
		S_\mu^{(1)}(E) &:=\,\sup\limits_{\Lambda'\subseteq\Lambda}\sup\limits_{\x\in\mathcal{C}_{\Lambda'}} 	
			\sum\limits_{\y\in\mathcal{C}_{\Lambda'}}e^{s\mu \vert x_1-y_1\vert} \,
			\mathbb{E}\bigl[\vert G_{\Lambda'}(\x,\y;E)\vert^s\bigr]\leqslant C^{(1)}_s\,,\label{definition_S1}\\
		S_\mu^{(2)}(E) &:=\,\sup\limits_{\Lambda'\subseteq\Lambda}\sup\limits_{\x\in\mathcal{C}_{\Lambda'}}
			\sum\limits_{\y\in\mathcal{X}_{\Lambda'}^{n}}e^{s\mu \overline{d}(\x,\y)} \,
			\mathbb{E}\bigl[\vert G_{\Lambda'}(\x,\y;E)\vert^s\bigr]\leqslant C^{(2)}_s\,,\label{definition_S2}
	\end{align}
	where $\mathcal{C}_\Lambda\equiv\mathcal{C}_\Lambda^{(1)}$ abbreviates clustered configurations in $\chiln$.
\end{theorem}
The proof of Theorem \ref{theorem_exponential_fm} mainly relies on Theorem~\ref{theorem_combes_thomas}, some properties of $\overline{d}$ and a standard resolvent expansion, which for the convenience of the reader we summarize in:
\begin{proposition}\label{lemma_resolvent_expansion}
	Let $H$ be a Hamiltonian on some separable Hilbert space with orthonormal basis $\lbrace\delta_x\rbrace$ and let $(Q,P)$ be a pair of non-trivial complementary orthogonal projections on $\mathcal{H}$. Then, for any $\delta_x\in\mathrm{ran}(Q)$, $\delta_y\in\mathrm{ran}(P)$ and $E\in\mathbb{C} \backslash \sigma(H) $, we have
	\begin{equation}\label{eq:reseq}
		\vert G(x,y)\vert
			\leqslant \sum\limits_{\substack{\delta_u\in\mathrm{ran}(Q)\\ \delta_v\in\mathrm{ran}(P)}}
			\vert G(x,u)\vert\vert H(u,v)\vert\vert\langle\delta_v, (P(H-E)P)^{-1}\delta_y\rangle\vert\,,
	\end{equation}
	where $G(x,y) := \langle\delta_x,(H-E)^{-1}\delta_y\rangle$ and $H(x,y):=\langle\delta_x,H\delta_y\rangle$.
\end{proposition}
\begin{remark}
	Using the resolvent identity and the Combes-Thomas estimate once more, one may also show that for any $\x,\y\in\chiln$:
	\begin{equation}\label{exponential_fm}
		\mathbb{E}\bigl[\vert G_\Lambda(\x,\y;E)\vert^s\bigr]\leqslant C_s \, e^{-s\mu D(\x,\y)} \, , 
	\end{equation}
	where
	\begin{equation}\label{definition_D}
		D(\x,\y):=\min\bigl\lbrace d(\x,\y),\overline{d}(\x,\y)\bigr\rbrace\,.
	\end{equation}
	In contrast to $\overline{d}$, $D$ is a distance function, as shown in Lemma~\ref{lemma_distance_dbar}(i).
\end{remark}
We are now ready to give a 
\begin{proof}[Proof of Theorem \ref{theorem_exponential_fm}]
	The dependence of quantities on the energy $E$ shall be omitted throughout the proof. Nevertheless, note that
	\begin{equation*}
		C_\textsc{t}(I):= \sup\limits_{E\in I} C_\textsc{t}(g,\mu_\textsc{t},E) =C_\textsc{t}(g,\mu_\textsc{t},\sup I)\in (0,\infty).
	\end{equation*}
	Let first $\x,\y\in\mathcal{C}$. Such configurations satisfy either $x_1<y_1$, $\x=\y$ or $x_1>y_1$. In the middle case, the summand is bounded using~\eqref{assumption_a2}, while the remaining two cases may by symmetry be treated in a similar manner. We thus only present the case $x_1<y_1$. Let $\Lambda_\x := \Lambda\cap[x_1+1,\infty)$ and $\mathcal{X}_\x$ ($\mathcal{C}_\x$) be the subset of (clustered) configurations to the right of $x_1$, i.e.
	\begin{equation*}
		\mathcal{X}_\x := \left\lbrace \textbf{z}\in\chiln : x_1 < z_1\right\rbrace\equiv\mathcal{X}_{\Lambda_\x}^n\,,\qquad
		\mathcal{C}_\x := \mathcal{C}\cap\mathcal{X}_\x\,.
	\end{equation*}
	Let $P_\x$ be the projection onto $\mathcal{H}_{\Lambda_\x}^n\equiv l^2(\mathcal{X}_\x)$ and $Q_\x$ its orthogonal complement in $\hln$. Applying the resolvent equation~(\ref{eq:reseq}), we obtain the expansion
	\begin{equation*}
		\vert G_{\Lambda}(\x,\y)\vert \leqslant 
			\sum\limits_{\substack{\textbf{u}\notin\mathcal{X}_\x\\ \textbf{v}\in\mathcal{X}_\x}}
			\vert G_{\Lambda}(\x,\textbf{u})\vert \vert H_\Lambda(\textbf{u},\textbf{v})\vert 
			\vert G_{\Lambda_\x}(\textbf{v},\y)\vert\,.
	\end{equation*}
	The matrix element $H_\Lambda(\textbf{u},\textbf{v})$ is only non-vanishing for $\textbf{u}=\textbf{v}$ or $d(\textbf{u},\textbf{v})=1$. Combined with the imposed conditions, $\textbf{u}\notin\mathcal{X}_\x$ and $\textbf{v}\in\mathcal{X}_\x$, the sum is restricted to pairs $\lbrace\textbf{u},\textbf{v}\rbrace$ with $u_1=x_1$, $v_1=x_1+1$, and $u_k=v_k$, $2\leqslant k\leqslant n$. Hence, for any such $\textbf{v}$, there is a unique neighboring $\textbf{u}\equiv\textbf{u}( \textbf{v}) $. Only one such $\textbf{v}$ is also fully clustered; we shall denote it by $\textbf{v}_0\in\mathcal{C}$. We may thus bound $\vert G_{\Lambda}(\x,\y)\vert$ by
	\begin{equation}\label{resolvent_expansion_1}
		\vert G_\Lambda(\x,\textbf{u}(\textbf{v}_0))\vert \vert G_{\Lambda_\x}(\textbf{v}_0,\y)\vert
			+\sum\limits_{\substack{\textbf{v}\notin\mathcal{C}\\v_1=x_1+1}}
			\vert G_\Lambda(\x,\textbf{u}(\textbf{v}))\vert\vert G_{\Lambda_\x}(\textbf{v},\y)\vert\,.
	\end{equation}
	By Proposition~\ref{lemma_resolvent_expansion} applied to the complementary projections $(\mathcal{P}^{(1)}, \mathcal{Q}^{(2)})$, the last resolvent factor, $\vert G_{\Lambda_\x}(\textbf{v},\y)\vert$ with $\y\in\mathcal{C}$, may now be expanded into
	\begin{align}
		\vert G_{\Lambda_\x}(\textbf{v},\y)\vert
			&\leqslant\, \sum\limits_{\substack{\textbf{w}\in\mathcal{X}_\x\setminus\mathcal{C}_\x\\ 	
				\textbf{z}\in\mathcal{C}_\x}} \vert G_{\Lambda_\x}^{(2)}(\textbf{v},\textbf{w})\vert
				\vert H_{\Lambda_\x}(\textbf{w},\textbf{z})\vert \vert G_{\Lambda_\x}(\textbf{z},\textbf{y})\vert \notag\\
			&\leqslant\, 2\, C_\textsc{t} e^{\mu_\textsc{t}}\sum\limits_{\textbf{z}\in\mathcal{C}_\x}
				e^{-\mu_\textsc{t}d(\textbf{v},\textbf{z})}\vert G_{\Lambda_\x}(\textbf{z},\textbf{y})\vert\,.
				\label{inequality_resolvent_ct}
	\end{align}
	The second inequality is by the Combes-Thomas estimate (\ref{combes_thomas}) and the observation that, for any $\w\not\in\mathcal{C}_\x$, either $H_\Lambda(\textbf{w},\textbf{z})= 0$ or $d(\textbf{w},\textbf{z})=1$ and $\vert H_\Lambda(\textbf{w},\textbf{z})\vert = 1$. The factor $2$ follows from the fact that, for any $\textbf{z}\in\mathcal{C}$, at most two $\textbf{w}$ satisfy $d(\textbf{w},\textbf{z})=1$.
	
	The resolvent $G_{\Lambda_\x}$ is independent of $\omega(x_1)$. Taking the fractional moment of (\ref{resolvent_expansion_1}), conditioning on the random potential at all sites but $x_1$, and using (\ref{assumption_a2}) (with $ u=v = x_1$) thus yields:
	\begin{align*}
		\mathbb{E}\bigl[\vert G_\Lambda(\x,\y)\vert^s\bigr]
			\leqslant& \frac{{C}_s}{\lambda^s} \, \mathbb{E}\bigl[\vert G_{\Lambda_\x}(\textbf{v}_0,\y)\vert^s\bigr]\\
			&+\, \bigl(2C_\textsc{t}e^{\mu_\textsc{t}}\bigr)^s\frac{{C}_s}{\lambda^s} 
			\sum\limits_{\substack{\textbf{v}\notin\mathcal{C}\\v_1=x_1+1}}\sum\limits_{\textbf{z}\in\mathcal{C}_\x}
			e^{-s\mu_\textsc{t}d(\textbf{v},\textbf{z})}\mathbb{E}\bigl[\vert G_{\Lambda_\x}(\textbf{z},\y)\vert^s\bigr]\,.
	\end{align*}
	By $ | x_1-y_1 | = 1+ |v_1-y_1 | \leqslant 1+  |v_1-z_1| + |z_1-y_1 | $ and Lemma \ref{lemma_sum}, we then obtain for any $\mu<\mu_\textsc{t}$ the bound
	\begin{align*}
		\sum\limits_{\y\in\mathcal{C}_{\Lambda_\x}} e^{s\mu\vert x_1-y_1\vert}\,
			\mathbb{E}\bigl[\vert G_\Lambda(\x,\y)\vert^s\bigr]
		&\leqslant \frac{{C}_s}{\lambda^s}e^{s\mu} S^{(1)}_\mu + 
			\bigl(2C_\textsc{t}e^{\mu_\textsc{t}}\bigr)^s  \frac{{C}_s}{\lambda^s} \, e^{s\mu}  \notag \\
		& \mkern-230mu\times \sum\limits_{\substack{\textbf{v}\notin\mathcal{C}\\v_1=x_1+1}}
			\sum\limits_{\textbf{z}\in\mathcal{C}_\x}  e^{-s(\mu_\textsc{t}-\mu) |v_1-z_1| }
			e^{-s \mu_\textsc{t} \sum_{j=2}^n |v_j-z_j| }
			\sum\limits_{\y\in\mathcal{C}_{\Lambda_\x}} e^{s\mu\vert z_1-y_1\vert}
			\bigl[\vert G_{\Lambda_\x}(\z,\y)\vert^s\bigr]  \notag \\
		&\leqslant  \frac{{C}_s}{\lambda^s}e^{s\mu}  \left( 1 + \bigl(2C_\textsc{t}e^{\mu_\textsc{t}}\bigr)^s
			\frac{C_\infty(s\mu_\textsc{t})}{1-e^{-s(\mu_\textsc{t}-\mu)}} \right) \, S^{(1)}_\mu\,,
	\end{align*}
	which in turn implies the following estimate of (\ref{definition_S1}):
	\begin{equation*}
		S^{(1)}_\mu \leqslant \frac{{C}_s}{\lambda^s} + 2\frac{{C}_s}{\lambda^s}e^{s\mu}
			\left(1+\bigl(2C_\textsc{t}e^{\mu_\textsc{t}}\bigr)^s
			\frac{C_\infty(s \mu_\textsc{t})}{1-e^{-s(\mu_\textsc{t}-\mu)}}\right) S^{(1)}_\mu \,.
	\end{equation*}
	The first term arises from the case $\y=\x$ and the factor $2$ accounts for the two cases $x_1<y_1$ and $x_1>y_1$. Hence, if $\lambda$ is chosen so large that the coefficient of $S^{(1)}_\mu $ on the right hand side is strictly less than $1$, $S^{(1)}_\mu$ (which is finite) is uniformly bounded as claimed.
	\medskip\\
	Let $\x\in\mathcal{C}$ and $\y\in\mathcal{X}$. We have shown that the terms with clustered $\y\in\mathcal{C}_{\Lambda'}$ in (\ref{definition_S2}) are bounded as in (\ref{definition_S1}), and therefore we henceforth concentrate on the case $\x\in\mathcal{C}_{\Lambda'}$ and $\y\notin\mathcal{C}_{\Lambda'}$. Applying Proposition~\ref{lemma_resolvent_expansion} with the pair of complementary projections $(\mathcal{P}^{(1)},\mathcal{Q}^{(2)})$ and in turn the Combes-Thomas estimate (\ref{combes_thomas}), we have the following upper bound for $\vert G_\Lambda(\x,\y)\vert$:
	\begin{equation*}
		\sum\limits_{\substack{\textbf{u}\in\mathcal{C}_{\Lambda}\\ \textbf{v}\not\in\mathcal{C}_{\Lambda}}}
		\vert G_\Lambda(\x,\textbf{u})\vert \vert H_\Lambda(\textbf{u},\textbf{v})\vert
		\vert G_\Lambda^{(2)}(\textbf{v},\y)\vert
		\leqslant 2C_\textsc{t}e^{\mu_\textsc{t}}\sum\limits_{\textbf{u}\in\mathcal{C}_\Lambda}
			e^{-\mu_\textsc{t} d(\textbf{u},\y)}\vert G_\Lambda(\x,\textbf{u})\vert.
	\end{equation*}
	As in (\ref{inequality_resolvent_ct}), the second inequality uses that non-vanishing contributions have $d(\textbf{u},\textbf{v})=1$ and $\vert H_\Lambda(\textbf{u},\textbf{v})\vert = 1$, and the factor $2$ reflects the fact that at most two $\textbf{v}$ satisfy $d(\textbf{u},\textbf{v})=1$. By (\ref{definition_S1}) and Lemma \ref{lemma_distance_dbar}(iv) this yields:
	\begin{align*}
		&\sup\limits_{\Lambda'\subseteq\Lambda}\sup\limits_{\x\in\mathcal{C}_{\Lambda'}}
			\sum\limits_{\y\not\in\mathcal{C}_{\Lambda'}}
			e^{s\mu \overline{d}(\x,\y)}\mathbb{E}\bigl[\vert G_{\Lambda'}(\x,\y)\vert^s\bigr]  \notag\\
		&\:\leqslant\,\sup\limits_{\Lambda'\subseteq\Lambda}\sup\limits_{\x\in\mathcal{C}_{\Lambda'}}
			\bigl(2C_\textsc{t}e^{\mu_\textsc{t}}\bigr)^s
			\sum\limits_{\y\not\in\mathcal{C}_{\Lambda'}}\sum\limits_{\textbf{u}\in\mathcal{C}_{\Lambda'}}
			e^{s\mu\overline{d}(\textbf{u},\y)-s\mu_\textsc{t}d(\textbf{u},\y)}e^{s\mu\overline{d}(\x,\textbf{u})}
			\mathbb{E}\bigl[\vert G_{\Lambda'}(\x,\textbf{u})\vert^s\bigr] \notag\\
		&\:\leqslant\, \bigl(2C_\textsc{t}e^{\mu_\textsc{t}}\bigr)^s C_\infty\bigl(s(\mu_\textsc{t}-\mu)\bigr) \, C^{(1)}_s\,,
	\end{align*}
	where the uniform boundedness of $C_\infty$ is by Lemma \ref{lemma_sum}.
\end{proof}

\section{From fractional moments to eigenfunction correlators}\label{section_fm_to_eq}

As explained in \cite{A94,AWbook}, in order to relate the configurational eigenfunction correlator to the fractional moment of the Green function, it is useful to consider the family of interpolated eigenfunction correlators with parameter $ s \in [0, 1] $:
\begin{equation*}
	Q^{(n)}_\Lambda(\x,\y, I , s) := \sum_{E \in \sigma(H_\Lambda)\cap I}  
		\left| \langle \delta_\x , P_{\{E\}}(H_\Lambda) \delta_\x \rangle \right|^{1-s}
		\left| \langle \delta_\x , P_{\{E\}}(H_\Lambda) \delta_\y \rangle \right|^s \,.
\end{equation*}
The following bounds, which are taken from~\cite[Ch.~7.3.2]{AWbook} (see also \cite{AiW2009}), hold for any $s \in [0,1] $: 
\begin{equation}\label{eq:ECbounds}
	Q^{(n)}_\Lambda(\x,\y, I ) \leqslant \sqrt{ Q^{(n)}_\Lambda(\x,\y,I,s) \, Q^{(n)}_\Lambda(\y,\x,I,s)} \,, \quad 
  	Q^{(n)}_\Lambda(\x,\y,I,s)  \leqslant 1 \, . 
\end{equation}

In \cite{AiW2009} a general relation concerning the eigenfunction correlator was derived which in our situation reads: 
\begin{proposition}[Thm.~4.5 in \cite{AiW2009}] \label{prop:EC}
	Let $s\in(0,1)$, $ \lambda > 0 $ and  $ I \subset \mathbb{R} $ be an interval. Under Assumptions~\textbf{A1}-\textbf{A2}, there exists a constant $ C \in (0,\infty) $ such that for all $ n \in \mathbb{N}$, $ \Lambda $ and all $ \x, \y \in \mathcal{X}_{\Lambda}^{n}  $ with  $ u \in \x\,$:
	\begin{equation}
		\mathbb{E}\left[ Q^{(n)}_\Lambda(\x,\y,I,s) \right]
			\leqslant C \sum_{\substack{\w \in\mathcal{X}_{\Lambda}^{n} \\ u \in \w} }
			\int_I \mathbb{E}\left[ |G_\Lambda(\w,\y;E)|^s   \right]\, \mathrm{d}E \,. 
	\end{equation}
\end{proposition}
The proposition will be used to conclude localization bounds in case $ \x, \y \in \mathcal{C} $. The remaining case $\x \not\in \mathcal{C} $ and $ \y \in  \mathcal{X}_{\Lambda}^{n} $ is dealt with perturbatively. 

\begin{lemma}\label{lem:perturbationth}
	In the set-up of Theorem~\ref{theorem_exponential_cor}, for any $\x\not\in \mathcal{C}$ and $\y\in\mathcal{X}_{\Lambda}^{n}$:
	\begin{equation}
		Q^{(n)}_\Lambda(\x,\y,I) \leqslant 2\,C_\textsc{t}(I)\,e^{\mu_\textsc{t}} \sum_{\w \in \mathcal{C}}
			e^{- \mu_\textsc{t}  d(\x,\w) }  \, Q^{(n)}_\Lambda(\w,\y, I ) \,.
	\end{equation}
\end{lemma}
\begin{proof}
	The inequality is based on the following singular relation \cite[Prop.~7.9]{AWbook} between the eigenfunction correlator and the Green function:
	\begin{equation}\label{eq:singlim}
		Q_\Lambda^{(n)}(\x,\y, I )  = \lim_{s \nearrow 1} \frac{1-s}{2} \int_I |G_\Lambda(\x,\y; E)|^s \, \mathrm{d}E \,. 
	\end{equation}
	Using $ \x \not\in  \mathcal{C} $, the resolvent identity (\ref{eq:reseq}) and the Combes-Thomas estimate~\eqref{combes_thomas}  yield the following upper bound for $|G_\Lambda(\x,\y; E)|^s\,$:
	\begin{align*}
		|G_\Lambda^{(2)}(\x,\y;& E)|^s \, \mathds{1}_{\{\y\not\in\mathcal{C}\}}
	  		+\sum\limits_{\substack{\textbf{u} \not\in\mathcal{C}\\ \textbf{v}\in\mathcal{C}}} 
	  		|G_\Lambda^{(2)}(\x,\textbf{u}; E)|^s  |H_\Lambda(\textbf{u},\textbf{v})|^s
	  		| G_\Lambda(\textbf{v},\y;E) |^s \notag \\
	  	& \leqslant  C_\textsc{t}^s   e^{- s \mu_\textsc{t}  d(\x,\y) }\, + 2 \bigl(C_\textsc{t}e^{\mu_\textsc{t}}\bigr)^s
	  		\sum_{\textbf{v}\in \mathcal{C}} e^{- s \mu_\textsc{t} d(\x, \textbf{v})} | G_\Lambda(\textbf{v},\y,E) |^s \, . 
	\end{align*}
	The factor $2$ on the right side is due to the fact that for a given $ \textbf{v} \in   \mathcal{C} $ there are at most two configurations $ \textbf{u} \not\in   \mathcal{C} $ for which $ 0 \neq  |H_\Lambda(\textbf{u},\textbf{v})| = 1 $. In the singular limit~\eqref{eq:singlim}, the first term on the right side vanishes and the last term yields the claim.
\end{proof}

We are now ready to prove our main result:

\begin{proof}[Proof of Theorem~\ref{theorem_exponential_cor}]
	We treat the three cases 1.~$\x, \y \in \mathcal{C} $,  2.~$ \x \not\in\mathcal{C} , \;  \y \in \mathcal{C} $ and 3.~$  \x, \y \not\in\mathcal{C} $ separately.
	\medskip\\
	\noindent 1.~For $\x,\y\in\mathcal{C}$, we assume without loss of generality that $x_1\leqslant y_1 $. The bound~\eqref{eq:ECbounds} and a H\"older estimate yield for any $s\in [0,1]$:
	\begin{equation*}
		\mathbb{E}\left[Q^{(n)}_\Lambda(\x,\y,I)\right]
			\leqslant \mathbb{E}\left[Q^{(n)}_\Lambda(\x,\y,I,s)\right]^\frac{1}{2}
		\mathbb{E}\left[Q^{(n)}_\Lambda(\y,\x,I,s)\right]^\frac{1}{2}\,.
	\end{equation*}
	Now, Proposition~\ref{prop:EC} (with $ u = x_1 $) together with Lemma \ref{lemma_distance_dbar}(v) (by which $ y_1 - x_1 \leqslant \overline{d}(\w,\y) $ for any $ \w \ni x_1 $) implies:
	\begin{align*}
		\mathbb{E}\left[Q^{(n)}_\Lambda(\x,\y,I,s)\right] 
			&\leqslant\,C \, e^{-s\mu|x_1-y_1|} \int_I   \sum_{\substack{\w \in\mathcal{X}_{\Lambda}^{n} \\ x_1 \in \w} }
				e^{s\mu  \overline{d}(\w,\y)} \, \mathbb{E}\left[ |G_\Lambda(\w,\y;E)|^s   \right]\, \mathrm{d}E \notag \\
			&\leqslant C\, |I | \,\sup_{E \in I } S^{(2)}_\mu(E) \, e^{-s\mu|x_1-y_1|} \, .
	\end{align*}
	The finiteness of $\sup_{E \in I } S^{(2)}_\mu(E)\leqslant C_s^{(2)} $ is by \eqref{definition_S2}. Exchanging $\x$ and $\y$, and using the second statement of Lemma \ref{lemma_distance_dbar}(v), the same bound holds since $\vert x_n-y_n\vert=\vert x_1-y_1\vert$ for clustered $\x$ and $\y$. Hence, we obtain
	\begin{equation}\label{eq:ECcluster}
		\mathbb{E}\left[ Q^{(n)}_\Lambda(\x,\y,I)  \right]
			\leqslant  C\, |I | \, C_s^{(2)} \, e^{-s\mu|x_1-y_1|}\,.
	\end{equation}
	\noindent 2.~For $\x \not\in\mathcal{C}$ and $ \y \in \mathcal{C} $, we use Lemma~\ref{lem:perturbationth} together with the previous estimate~\eqref{eq:ECcluster} to conclude
	\begin{equation*}
		\mathbb{E}\left[ Q^{(n)}_\Lambda(\x,\y, I)  \right]
			\leqslant C\sum_{\w \in \mathcal{C} }  e^{- \mu_\textsc{t}  d(\x,\w) } e^{-s\mu|w_1-y_1|} \, . 
	\end{equation*}
	\noindent 3.~In the remaining case $\x, \y \not\in\mathcal{C}$, we use an iterated version of Lemma~\ref{lem:perturbationth}:
	\begin{equation*}
		Q^{(n)}_\Lambda(\x,\y, I )
			\leqslant (2 \, C_\textsc{t}(I) e^{\mu_\textsc{t}})^2 \sum_{\w , \textbf{v} \in \mathcal{C} }
			e^{- \mu_\textsc{t}  (d(\x,\w)+ d(\textbf{v},\y)) }  \, Q^{(n)}_\Lambda(\w,\textbf{v} , I ) \, . 
	\end{equation*}
	Averaging over the disorder and inserting~\eqref{eq:ECcluster} thus yields
	\begin{equation*}
		\mathbb{E}\left[ Q^{(n)}_\Lambda(\x,\y, I)  \right]
			\leqslant C \sum_{\w , \textbf{v} \in \mathcal{C} }
				e^{- \mu_\textsc{t}  (d(\x,\w)+ d(\textbf{v},\y)) }    e^{-s\mu|w_1-v_1|} \, . \qedhere 
	\end{equation*}
\end{proof}

\appendix

\section{Proof of thresholds}\label{app_thresholds}

The proof of Lemma \ref{lemma_thresholds} requires some preliminary considerations. Let $\alpha\neq\beta$, $\alpha, \beta\in\laml$ be two sites and the operators $\pi_\alpha$ and $\pi_{\alpha,\beta}$ be defined on $\hln$ through
\begin{align*}
	\pi_\alpha\,\delta_\x :=\,\begin{cases}\delta_\x\,,&\textnormal{if }\alpha\in\x\,,\\
		0\,,&\textnormal{otherwise},\end{cases}\qquad
	\pi_{\alpha,\beta}\,\delta_\x &:=\,\begin{cases}\delta_\x\,,
		&\textnormal{if either }\alpha\in \x\textnormal{ or }\beta\in \x\,,\\ 0\,,&\textnormal{otherwise.}\end{cases}
\end{align*}
It is worthwhile noting that a basis element $\delta_\x$ is only in the range of $\pi_{\alpha,\alpha+1}$ if the configuration $\x$ has a cluster ending on $\alpha$ or beginning on $\alpha+1$. The operators $\pi_{\alpha,\alpha+1}$ are related to transitions between occupied and empty sites.
A second operator $\tau_{\alpha,\beta}$ is defined on $\laml^n$ as
\begin{align*}
	\tau_{\alpha,\beta}(\x)_i :=\begin{cases}\alpha\,,&\textnormal{if }x_i = \beta\,,\\ \beta\,,
		&\textnormal{if }x_i = \alpha\,,\\ x_i\,,&\textnormal{otherwise}\,.\end{cases}
\end{align*}
After reordering, this also defines an operator on $\chiln$ which we denote by the same symbol. The action of $\tau_{\alpha,\beta}$ then amounts to exchanging the sites $\alpha$ and $\beta$ together with their occupancy. In particular, $\tau_{\alpha,\beta}(\x)=\x$ if $\alpha, \beta\in \x$ or $\alpha, \beta\notin \x$. By the embedding $\x\mapsto\delta_\x$, $\tau_{\alpha,\beta}$ is extended to an operator on $\hln$ through $\tau_{\alpha,\beta}\,\delta_\x := \delta_{\tau_{\alpha,\beta}(\x)}$. In the special case $\beta=\alpha+1$, the operator $\tau_{\alpha,\alpha+1}$ describes the hopping between the neighboring sites $\alpha$ and $\alpha+1$.
\begin{proposition}\label{proposition_estimates_UA}
	Let $A$ and $U$ be defined as in \eqref{adjacency} resp.~\eqref{potential} on $ \hln $ and set $\laml^{-} := \left[-L,L-1\right]\cap\mathbb{Z}$. Then:
	\begin{enumerate}[label=(\roman*)]
		\item $2\,U = \sum\limits_{\alpha\in\laml^{-}}\pi_{\alpha,\alpha+1} + \pi_{-\textsc{l}}+\pi_\textsc{l}$,
		\item $-A = \sum\limits_{\alpha\in\laml^{-}}\bigl(1 -\pi_{\alpha,\alpha+1} - \tau_{\alpha,\alpha+1}\bigr)$.
	\end{enumerate}
\end{proposition}
\begin{proof}
	Each claim is established by proving that both sides of the equalities coincide when applied to the basis vectors $\delta_\y$, $\y\in\chiln$.
	\medskip\\
	\noindent (i) From the preliminary considerations above, $\sum_{\alpha\in\laml^{-}}\pi_{\alpha,\alpha+1}\delta_\y = C_\y\,\delta_\y$, where $C_\y$ is the number of transitions in $\y$ between occupied and empty sites; or equivalently the number of other configurations that can be obtained from $\y$ by moving one particle to an adjacent empty site. Comparing with the definition (\ref{potential}) of $U$ shows that $2\,U\delta_\y = \bigl(C_\y+\pi_{-\textsc{l}}+\pi_\textsc{l}\bigr)\,\delta_\y$.
	\medskip\\
	\noindent (ii) We have
	\begin{align*}
		\left(1-\pi_{\alpha,\alpha+1} - \tau_{\alpha,\alpha+1}\right)\delta_\y 
			= \begin{cases}-\tau_{\alpha,\alpha+1}\delta_\y\,,&\textnormal{if either }\alpha\in\y
			\textnormal{ or }\alpha+1\in\y\,,\\ 0\,,&\textnormal{otherwise}.\end{cases}
	\end{align*}
	Comparing with the definition (\ref{adjacency}) of $A$, we obtain the claimed identity. In fact, any neighboring pair $\lbrace \x,\y\rbrace $ is uniquely written as $\lbrace \tau_{\alpha,\alpha+1}(\y),\y\rbrace$ for some $\alpha\in\laml^{-}$ with either $\alpha\in\y$ or $\alpha+1\in\y$.
\end{proof}
\begin{proof}[Proof of Lemma \ref{lemma_thresholds}]
	By assumption, $\lambda V\geqslant 0$. From Proposition \ref{proposition_estimates_UA}, we have
	\begin{equation*}
		H_\Lambda \geqslant \sum\limits_{\alpha\in\laml^{-}}\left[\bigl(1-\tau_{\alpha,\alpha+1}\bigr)
			+ (g-1)\pi_{\alpha,\alpha+1}\right] + g\bigl(\pi_{-\textsc{l}}+\pi_\textsc{l}\bigr)\,.
	\end{equation*}
	By their definition, $\tau_{\alpha,\beta}$ are hermitian and unitary so that $1-\tau_{\alpha,\beta}\geqslant 0$ and hence
	\begin{equation*}
		H_\Lambda \geqslant (g-1)\sum\limits_{\alpha\in\laml^{-}}\pi_{\alpha,\alpha+1} 
			+ g\bigl(\pi_{-\textsc{l}}+\pi_\textsc{l}\bigr)\,.
	\end{equation*}
	As said, $\sum_{\alpha\in\laml^{-}}\pi_{\alpha,\alpha+1}$ counts the number of transitions between occupied and empty sites. A configuration in $\mathrm{ran}(\mathcal{Q}^{(k)})$ features either $2(k-1)$, $2k-1$ or $2k$ such transitions, depending on whether it encompasses two, one or no boundary clusters. Each boundary cluster being however penalized by $g\bigl(\pi_{-\textsc{l}}+\pi_\textsc{l}\bigr)$, the configurations with no boundary clusters are energetically most favorable, establishing the lower bound on $\mathcal{Q}^{(k)}H_\Lambda \mathcal{Q}^{(k)}$ at $2k(g-1)$.
\end{proof}

\section{Distance functions and summability}\label{app_auxiliary_results}

This appendix is dedicated to proving three technical, but important lemmas on properties of the (distance) functions $d$, $\overline{d}$ and $D =\min\lbrace d,\overline{d}\rbrace$ as defined in~\eqref{eq:el1}, \eqref{definition_dbar_app} and \eqref{definition_D}.
\begin{lemma}\label{lemma_distance_dbar}
	\begin{enumerate}[label=(\roman*)]
		\item $D(\cdot,\cdot)$ is a distance function on $\chiln$,
		\item $\overline{d}(\x,\y)=D(\x,\y)=\vert x_1-y_1\vert$ for all $\x,\y\in\mathcal{C}$,
		\item $\overline{d}(\x,\y)= D(\x,\y)=\min_{\textbf{v}\in\mathcal{C}}\lbrace\vert x_1-v_1\vert +d(\textbf{v},\y)\rbrace$ for all $\x\in\mathcal{C}$, $\y\in\chiln$,
		\item $\overline{d}(\x,\y)\leqslant \overline{d}(\x,\textbf{u}) +\overline{d}(\textbf{u},\y)$ for all $\x,\textbf{u}\in\mathcal{C}$, $\y\in\chiln$,
		\item $\overline{d}(\w,\y)\geqslant y_1-x_1$ for all $\y\in\mathcal{C}$ and $ \w\in\chiln$ such that $x_1\leqslant y_1$ and $x_1\in\w$; and similarly, $\overline{d}(\w,\x)\geqslant y_n-x_n$ for all $\x\in\mathcal{C}$ and $\w\in\chiln$ such that $y_n\geqslant x_n$ and $y_n\in\w$,
		\item $\overline{d}\bigl(\x,\y\bigr)\geqslant \mathrm{dist}(U,V)-(n-1)$ for all $\x\cap U\neq\emptyset$ and $\y\cap V\neq\emptyset$ with disjoint $U,V\subset\Lambda$.
	\end{enumerate}
\end{lemma}
\begin{proof}
	To establish property (i), we extend the natural graph of $\chiln$ with $\ell^1$-distance by edges connecting pairs $(\x,\y)$ of clustered configurations with $\vert x_1-y_1\vert = 1 $. As the function minimizing the distance on this graph, $D$ is a distance function. Properties (ii) and (iii) are immediate, by definition or $\vert x_1-y_1\vert\leqslant d(\x,\y)$. Property (iv) follows from (i)-(iii).
	\medskip\\
	(v) We use (iii): $ \overline{d}(\w,\y) = \min_{\mathbf{v} \in \mathcal{C}} \lbrace d(\w,\mathbf{v}) + |v_1-y_1| \rbrace $. Since $x_1 \in \w $, there is some $ j $ such that $ d(\mathbf{v}, \w) \geqslant |v_j - x_1 | $. By distinguishing the three cases, $ v_1 > y_1 $, $  y_1 \geqslant v_1 \geqslant x_1 $ and $ x_1 > v_1 $, one concludes from $ x_1 \leqslant y_1 $ and $ v_1 \leqslant v_j $ that $  |v_j - x_1 |  + |v_1-y_1| \geqslant y_1 - x_1 $. The second statement of (v) follows by left-right symmetry.
	\medskip\\
	(vi) Let $1\leqslant j,k\leqslant n$ be such that $x_j \equiv u\in U$ and $y_k \equiv v\in V$. Then, for any $\w,\z\in\mathcal{C}$, we have by the triangle inequality
	\begin{align*}
		d(\x,\w) + \vert w_1-z_1\vert + d(\z ,\y)\bigr\rbrace
			&\geqslant\, \vert u-w_j\vert + \vert w_j-z_j\vert + \vert z_k - v\vert\\
			&\geqslant\, \vert u-v\vert -\vert z_j-z_k\vert\,.
	\end{align*}
	Here, we used that $\vert w_1-z_1\vert=\vert w_j-z_j\vert$ for any $1\leqslant j\leqslant n$. The claim thus follows by $\vert z_j-z_k\vert\leqslant n-1$.
\end{proof}

The next lemma is the key observation regarding summability.
\begin{lemma}\label{lemma_sum}
	For any $\mu>0$ and $n\in\mathbb{N}$, we have
	\begin{equation}\label{inequality_sum}
		\sup\limits_{\x\in\mathcal{C}^{(1)}_\mathbb{Z}}\sum\limits_{\textbf{v}\in\mathcal{X}^{n}_\mathbb{Z}}
			e^{-\mu d(\x,\textbf{v})}
			\leqslant \frac{1}{1-e^{-\mu}}\left(\prod\limits_{k=1}^{\infty}\frac{1}{1-e^{-k\mu}}\right)^2
			=:C_{\infty}(\mu)\,.
	\end{equation}
\end{lemma}
\begin{remark}
	The product in the parenthesis is known in partition theory as Euler's generating function evaluated at $x=e^{-\mu}$ and, as an instance of $q-$Pochhammer symbol, often written as $(e^{-\mu},e^{-\mu})^{-1}_\infty$.
\end{remark}
\begin{proof}
	Notice first that $C_\infty(\mu)$ is a well-defined strictly decreasing function of $\mu>0$. It diverges to $+\infty$ for $\mu\to 0$ and converges to $1$ for $\mu\to +\infty$.
	
	The sum on the left hand side of (\ref{inequality_sum}) is translation invariant in $\x\in\mathcal{C}^{(1)}_\mathbb{Z}$, which may therefore be chosen arbitrarily. Setting $x_k=k$ and subsequently substituting $y_k:= v_k-k$, $1\leqslant k\leqslant n$, we obtain the expression
	\begin{equation*}
		\sum\limits_{y_1\in\mathbb{Z}} e^{-\mu\vert y_1\vert}\sum\limits_{y_2=y_1}^{\infty} e^{-\mu\vert y_2\vert}\:\cdots
			\sum\limits_{y_n=y_{n-1}}^{\infty}e^{-\mu\vert y_n\vert}\,.
	\end{equation*}
	Notice that $\lbrace y_1,y_2,\ldots,y_n\rbrace$ is a non-decreasing sequence. Hence, there exists $0\leqslant j\leqslant n$ such that $y_1,\ldots,y_j\in\mathbb{R}_{<0}$ and $y_{j+1},\ldots,y_n\in\mathbb{R}_{\geqslant 0}$. For $j=0$ resp.~$j=n$, the former resp.~the latter set are considered empty. Reordering the terms in the sums according to their $j$ yields 
	\begin{equation*}
		\sum\limits_{j=0}^{n}
		\left( \sum\limits_{y_1\leqslant\ldots\leqslant y_{j}\leqslant -1} 
			e^{\mu (y_1+\ldots +y_{j})}\right)
		\left(\sum\limits_{0\leqslant y_{j+1}\leqslant\ldots\leqslant y_n}e^{-\mu (y_{j+1}+\ldots +y_n)}\right)\,.
	\end{equation*}
	The claimed result now follows by the geometric series. In fact the first parenthesis is bounded by $e^{-j\mu}\cdot(e^{-\mu},e^{-\mu})^{-1}_\infty$ and the second by $(e^{-\mu},e^{-\mu})^{-1}_\infty$.
\end{proof}

An immediate consequence of the last lemma is
\begin{lemma}\label{summability_F}
	Let $F_\mu$ be as in~\eqref{definition_F} and $U,V\subset \Lambda$ be two connected, disjoint subsets. For any $\mu>0$, there exists a constant $C_\mu\in (0,\infty)$ such that for all $ n \geqslant 2\,$, $\Lambda$:
	\begin{equation}
		\sum\limits_{\substack{\x\in\chiln\\ \x\cap U\neq\emptyset}}
		\sum\limits_{\substack{\y\in\chiln\\ \y\cap V\neq\emptyset}} F_\mu(\x,\y) \leqslant C_\mu(n+1)\,.
	\end{equation}
\end{lemma}
\begin{proof}
	Assume without loss of generality that $U<V$, that is $u<v$ for all $u\in U$ and $v\in V$. Let $u_\mathrm{max}$ be the maximal element of $U$ and $v_\mathrm{min}$ the minimal element of $V$. To simplify notation, let henceforth $\mathcal{C}_A$ be the clustered configurations with at least one particle in $A\subset\Lambda$. We shall split the sum into four terms, according to whether $\x$ and $\y$ are clustered or not, and bound them separately. The emphasis of this proof lies on its shortness rather than on optimal bounds.
	\medskip\\
	\noindent 1.~$\x,\y\in\mathcal{C}$. Here, we have $F_\mu(\x,\y) = \exp\bigl[-\mu\vert x_1-y_1\vert\bigr]$. Since by assumption $u_\mathrm{max}\leqslant v_\mathrm{min}-1$, the following bound holds:
	\begin{equation}\label{eq_x_y_in_C}
		\sum\limits_{\substack{x_1\leqslant u_\mathrm{max}\\ y_1\geqslant u_\mathrm{max}-(n-2)}}
			e^{-\mu\vert x_1 - y_1\vert} \leqslant \frac{e^{-\mu}}{(1-e^{-\mu})^2} + (n-1)\coth\left(\frac{\mu}{2}\right)\,.
	\end{equation}
	The first term on the right hand side is the contribution for $y_1 > u_\mathrm{max}$ and the second term an estimate of the remainder.
	\medskip\\
	\noindent 2.~$\x\in\mathcal{C}$, $\y\not\in\mathcal{C}$. We estimate the sum
	\begin{equation}\label{eq:sum_F_2}
		\sum\limits_{\x\in\mathcal{C}_U}
		\sum\limits_{\substack{\y\in\chiln\setminus\mathcal{C}\\ \y\cap V\neq\emptyset}}
		\sum\limits_{\w\in\mathcal{C}} e^{-\mu\vert x_1-w_1\vert} e^{-\mu d(\w,\y)}\,.
	\end{equation}
	For clustered $\x$ and $\w$, $| x_1-w_1 | = | x_k-w_k |  $ for any $ k $. Moreover, since $\x$ has a particle in $U$, and $\y$ a particle in $V$, we have the lower bound
	\begin{equation*}
		\vert x_1 - w_1\vert + d(\w,\y)
			\geqslant \frac{1}{2}\bigl(\vert x_1-w_1\vert + \mathrm{dist}(U,\w) + \mathrm{dist}(\w,V) + d(\w,\y)\bigr)\,.
	\end{equation*}
	In fact, for any $\x\in\mathcal{C}_U$, there exist  $u\in U$ and $1\leqslant k(\x)\leqslant n$ such that $\vert x_1-w_1\vert = \vert u-w_{k(\x)}\vert\geqslant \min\lbrace \vert u-w_j\vert : u\in U, 1\leqslant j\leqslant n\rbrace =:\mathrm{dist}(U,\w)$; and similarly, $d(\w,\y)\geqslant \mathrm{dist}(\w,V)$. Hence, we may decouple the sum in (\ref{eq:sum_F_2}) and arrive at the bound
	\begin{align}
		& \sum\limits_{\x\in\mathcal{C}_U}\sum\limits_{\substack{\y\in\chiln\setminus\mathcal{C}\\ \y\cap V\neq\emptyset}}
			\sum\limits_{\w\in\mathcal{C}}
			e^{-\frac{\mu}{2}\vert x_1-w_1\vert}e^{-\frac{\mu}{2}\bigl(\mathrm{dist}(U,\w) + \mathrm{dist}(\w,V)\bigr)}
			e^{-\frac{\mu}{2}d(\w,\y)}  \label{eq:sum_F_3}  \\
		&\leqslant\, \left(\sum\limits_{\w\in\mathcal{C}} e^{-\frac{\mu}{2}\bigl(\mathrm{dist}(U,\w) + \mathrm{dist}(\w,V)\bigr)}\right) 
		\sup\limits_{\w\in\mathcal{C}} \left( \sum\limits_{\x\in\mathcal{C}_U}e^{-\frac{\mu}{2}\vert x_1-w_1\vert} 
			\sum\limits_{\substack{\y\in\chiln\setminus\mathcal{C}\\ \y\cap V\neq\emptyset}}
			e^{-\frac{\mu}{2}d(\w,\y)}\right)\,. \notag
	\end{align}
	By Lemma \ref{lemma_sum} and the geometric series, the last term on the right side is bounded by $\coth(\mu/4)\,C_\infty(\mu/2)$. To bound the first parenthesis, note that there only exist clustered $\w$ with $\mathrm{dist}(U,\w) + \mathrm{dist}(\w,V)=0$ if $U$ and $V$ are not too far apart; and in the worst case scenario where $v_\mathrm{min}=u_\mathrm{max}+1$, there are at most $(n-1)$ such $\w$. By the geometric series, the first parenthesis on the right hand side of (\ref{eq:sum_F_3}) is thus bounded by $(n-2)+\coth(\mu/4)$. By symmetry of $F_\mu$, the case $\x\not\in\mathcal{C}$ and $\y\in\mathcal{C}$ satisfies the same bound.
	\medskip\\
	\noindent 3.~$\x,\y\not\in\mathcal{C}$. The sum to be bounded reads
	\begin{equation}
		\sum\limits_{\substack{\x\in\chiln\setminus\mathcal{C}\\ \x\cap U\neq\emptyset}}
		\sum\limits_{\substack{\y\in\chiln\setminus\mathcal{C}\\ \y\cap V\neq\emptyset}}
		\sum\limits_{\w,\z\in\mathcal{C}} e^{-\mu d(\x,\w)}e^{-\mu\vert w_1-z_1\vert}e^{-\mu d(\z,\y)}\,.
	\end{equation}
	As in item 2, the estimate relies on decoupling the sums by giving lower bounds on $d(\x,\w)$, $\vert w_1-z_1\vert$, $d(\z,\y)$ or sums thereof which are independent of one variable. Here, we have
	\begin{equation*}
		d(\x,\w) \geqslant \frac{1}{2}\bigl(d(\x,\w) + \mathrm{dist}(U,\w)\bigr)\,,\qquad
		d(\z,\y) \geqslant\frac{1}{2}\bigl(\mathrm{dist}(\z,V)+d(\z,\y)\bigr)\,.
	\end{equation*}
	Combined with $\vert w_1-z_1\vert + \frac{1}{2} \mathrm{dist}(\z,V)\geqslant \frac{1}{2} \vert w_1-z_1\vert  +  \frac{1}{2} \mathrm{dist}(\w,V)$, the same bounds as in item 2 apply and yield
	\begin{equation}
		\sum\limits_{\substack{\x\in\chiln\setminus\mathcal{C}\\ \x\cap U\neq\emptyset}}
		\sum\limits_{\substack{\y\in\chiln\setminus\mathcal{C}\\ \y\cap V\neq\emptyset}} F_\mu(\x,\y)
			\leqslant \bigl((n-2)+\coth\left(\frac{\mu}{4}\right)\bigr)\coth\left(\frac{\mu}{4}\right)
			C_\infty\left(\frac{\mu}{2}\right)^2\,,
	\end{equation}
	concluding the proof.
\end{proof}
\begin{remark}
	For $U$ and $V$ not disjoint, the factor $(n-1)$ appearing for instance on the right hand side of (\ref{eq_x_y_in_C}) is typically replaced by $\vert U\cap V\vert + (n-1)$ and the expression $C_\mu(n+1)$ then depends on the size of the overlap.
\end{remark}

\section{Proof of relation to XXZ}\label{app_XXZ_relation}

\begin{proof}[Proof of Proposition~\ref{prop:XXZ}]
	For $\mathbb{C}^2$, we introduce the basis of eigenvectors $\lbrace e^{\pm}\rbrace$ of $S^{z}$ satisfying $S^{z}e^{\pm} = \pm\frac{1}{2}e^{\pm} $. Basis elements of $\mathcal{H}_{\vert\Lambda\vert}^\textsc{xxz}$ are $\vert\Lambda\vert$-fold tensor products of $e^{\pm}$ and uniquely determined through the number and positions of their down-spins, which may be summarized in a configuration $\x\in\mathcal{X}_\Lambda^n$ for some appropriate $0\leqslant n\leqslant\vert\Lambda\vert$. Denoting any such basis element by $e_\x$, this induces a unitary operator
	\begin{equation*}
		\mathcal{U}: \mathcal{H}_{n,\vert\Lambda\vert}^{\textsc{xxz}}\to \mathcal{H}_\Lambda^n\,,\qquad e_\x\mapsto\delta_\x
	\end{equation*}
	between the superselection sector with exactly $n$ down-spins and the space of $n$ hard-core particles.
	
	Next, we define the ladder operators $S^{\pm} := S^{x}\pm iS^{y}$ and the down-spin number operator $N:=S^{-}S^{+}=1/2-S^{z}$. They satisfy the following relations: $ S^{+}e^{+}=0 $, $ S^{-}e^{+}=\,e^{-} $, $ N e^{+} = 0 $, $		S^{+}e^{-}=e^{+} $, $ S^{-}e^{-}=\,0$ , $ N e^{-} = e^{-} $. The right hand side of (\ref{equivalence_hamiltonians}) may be recast as
	\begin{equation*}
		2 H_{n,\vert\Lambda\vert}^\textsc{xxz}(g,1)
			= -\sum\limits_{k\in\Lambda^{-}}A_k
			+2g\left(\sum\limits_{k\in\Lambda^{-}}U_k + \frac{1}{2}(N_\textsc{-l}+N_\textsc{l})\right)
			+\frac{\lambda}{g} \sum\limits_{k\in\Lambda}\omega(k) \, N_k\,,
	\end{equation*}
	where
	\begin{equation*}
		A_k := S_k^{+}\otimes S_{k+1}^{-}+S_{k}^{-}\otimes S_{k+1}^{+}\,, \qquad U_k := \frac{1}{4}
			\cdot\mathds{1}-\bigl(N_k-\frac{1}{2}\mathds{1}\bigr)\otimes\bigl(N_{k+1}-\frac{1}{2}\mathds{1}\bigr)\,,
	\end{equation*}
	and $\Lambda^{-}=[-L,L-1]\cap\mathbb{Z}$. Observe that both $A_k$ and $U_k$ vanish on $e_k^{+}\otimes e_{k+1}^{+}$ and $e_k^{-}\otimes e_{k+1}^{-}$. Moreover, when applied on $e_k^{+}\otimes e_{k+1}^{-}$ or $e_k^{-}\otimes e_{k+1}^{+}$, $A_k$ exchanges the spins while $U_k$ counts $1/2$, reminiscent of the actions of $A$ and $U$. A notable difference lies in that $U$ counts $1$ per cluster, while $\sum_{k\in\Lambda^{-}} U_k$ counts $1/2$ per interface between up- and down-spins. The additional potential term $(N_\textsc{-l}+N_\textsc{l})/2$ accounts for the lacking interface whenever a cluster sits at a boundary. The claim now follows by verifying that the actions of $\sum_{k\in\Lambda^{-}} A_k$ and $\sum_{k\in\Lambda^{-}} U_k + (N_\textsc{-l}+N_\textsc{l})/2$ on basis elements $e_\x$ indeed coincide with that of $A$ and $U$ on $\delta_\x$.
\end{proof}


\subsection*{Acknowledgment}
 We thank B.~Nachtergaele for proposing the problem and for illuminating discussions during his stay at TUM as a John von Neumann Fellow. V.B.~was supported by Grant P2EZP2\_162235 of the Swiss National Science Foundation. The final publication is available at Annales Henri Poincar\'e, Springer via http://dx.doi.org/10.1007/s00023-017-0591-0

\bibliography{XXZ_May_11}
\bibliographystyle{abbrv}

\end{document}